\documentclass[11pt]{article}

\usepackage[ruled, vlined, linesnumbered]{algorithm2e}
\SetKw{Break}{break}
\SetKwData{Null}{null}
\usepackage{xcolor}

\usepackage{fullpage}
\usepackage{array}
\usepackage{diagbox}
\usepackage{multirow}
\usepackage{amsmath,amssymb}
\usepackage[colorlinks,linkcolor=blue,citecolor=blue,urlcolor=black]{hyperref}
\usepackage{dsfont}
\usepackage{caption, subcaption}
\usepackage{framed,url}

\newtheorem{theorem}{Theorem}[section]
\newtheorem{lemma}[theorem]{Lemma}
\newtheorem{claim}[theorem]{Claim}
\newtheorem{corollary}[theorem]{Corollary}

\newtheorem{fact}[theorem]{Fact}

\newtheorem{remark}[theorem]{Remark}

\newtheorem{definition}[theorem]{Definition}
\newtheorem{property}[theorem]{Property}

\newcommand{\abs}[1]{{\left | #1 \right |}}
\newcommand{\gray}[1]{\textcolor{gray}{#1}}
\newcommand{\E}{\mathbf{E}}
\newcommand{\Var}{\mathbf{Var}}
\newcommand{\eps}{\epsilon}

\newcommand{\A}{\mathcal{A}}
\renewcommand{\Pr}{\mathbf{Pr}}

\newcommand{\MC}{{\text{mult}}}
\newcommand{\K}{\mathcal{K}}
\newcommand{\Tail}{\text{tail}}

\newcommand{\kV}[1]{V^{(#1)}}
\usepackage{xparse}
\newcommand{\loc}{\text{local}}
\NewDocumentCommand{\dl}{m o}{%
	\IfNoValueTF{#2}
	{d_{#1,\loc}}
	{d_{#1,\loc}^{#2}}%
}
\NewDocumentCommand{\Dl}{m o}{%
	\IfNoValueTF{#2}
	{D^{(#1)}_{\loc}}
	{D_{#1,\loc}^{#2}}%
}

\SetKwData{Null}{null}

\newcommand{\kDISJ}{DISJ$_{n,k}$}

\newenvironment{proof}{\trivlist\item[]\emph{Proof}:}%
{\unskip\nobreak\hskip 1em plus 1fil\nobreak$\Box$
	\parfillskip=0pt%
	\endtrivlist}

\newif\iffull
\fulltrue      % comment this for short version
\begin{document}
	
	\title{Frequency Moments in Noisy Streaming and Distributed Data under Mismatch Ambiguity\thanks{K.~Liu is partly supported by NSF CCF-1844234; Q.~Zhang is partly supported by NSF CCF-1844234 and IU Luddy Faculty Fellowship.}}
	
	\author{
		Kaiwen Liu \\ 
		Computer Science Department\\
		Indiana University\\
		\texttt{kaiwliu@iu.edu}
		\and 
		Qin Zhang \\
		Computer Science Department\\
		Indiana University\\
		\texttt{qzhangcs@iu.edu}
	}

		\maketitle
	
	\begin{abstract}
		We propose a novel framework for statistical estimation on noisy datasets. Within this framework, we focus on the frequency moments ($F_p$) problem and demonstrate that it is possible to approximate $F_p$ of the unknown ground-truth dataset using sublinear space in the data stream model and sublinear communication in the coordinator  model, provided that the approximation ratio is parameterized by a data-dependent quantity, which we call the $F_p$-mismatch-ambiguity. We also establish a set of lower bounds, which are tight in terms of the input size.  Our results yield several interesting insights:
		\begin{itemize}
			\item[--] In the data stream model, the $F_p$ problem is inherently more difficult in the noisy setting than in the noiseless one. In particular, while $F_2$ can be approximated in logarithmic space in terms of the input size in the noiseless setting, any algorithm for $F_2$ in the noisy setting requires polynomial space.
			
			\item[--] In the coordinator model, in sharp contrast to the noiseless case, achieving polylogarithmic communication in the input size is generally impossible for $F_p$ under noise. However, when the $F_p$ mismatch ambiguity falls below a certain threshold, it becomes possible to achieve communication that is entirely independent of the input size.
		\end{itemize}
	\end{abstract}

	\section{Introduction}
	\label{sec:intro}
	
	The presence of near-duplicates has long been a challenge in data management. Such noise may arise from diverse sources, such as data collection errors, representational inconsistencies, reformatting or compression artifacts, as well as inherent physical limitations like the imprecision of scientific instruments or the probabilistic nature of quantum systems. For example,
	\begin{itemize}
		\item[--] A large number of queries submitted to a search engine or a large language model, where different combinations of keywords may convey the same semantic intent.
		
		\item[--] Webpages, documents or images with near-duplicates are often stored across multiple 
		servers, and our goal is to compute aggregate statistics over these objects.
		
		\item[--] If we want to store quantum states classically, then we have to tolerate near-duplicates~\cite{ZH25}.
	\end{itemize}
	Typically, a thorough data cleaning step is performed before analysis~\cite{HSW07,EIV07,KSS06,DN09}. However, this approach is infeasible in ``big data'' settings, where the objective is to compute statistical functions over the dataset without 
	storing or communicating it in its entirety.
	
	To address this challenge, a line of research aims to conduct statistical analysis directly on noisy datasets, mitigating the noise on the fly without explicitly cleaning the data~\cite{Zhang15,CZ16,CZ18,Zhang25}.  However, as we shall discuss in the related work (Section~\ref{sec:related}), existing approaches to statistical estimations on noisy datasets exhibit inherent limitations. In this work, we introduce a new framework for analyzing noisy datasets. We then study the $p$-th frequency moments problem ($F_p$) within two established big data models: 
	\begin{itemize}
		\item[--] The data stream model~\cite{AMS99}: We sequentially process the items in input dataset $\sigma = (\sigma_1, \ldots, \sigma_m)$ in one or multiple passes to compute $F_p(\sigma)$. The main goal is to minimize memory space usage.
		
		\item[--] The coordinator model~\cite{PVZ12}: We have $k$ sites and one central coordinator.  The dataset $\sigma$ is partitioned across the $k$ sites, with the $i$-th site holding $\sigma^{(i)}$.  The $k$ sites and coordinator want to jointly compute $F_p(\sigma)$ via communication.  The computation proceeds in rounds. In each round, the coordinator sends a message to each of the $k$ sites, which then respond to the coordinator with a message. At the end, the coordinator outputs the answer. We would like to minimize both the communication cost and the number of rounds of the computation.
	\end{itemize}
	We design algorithms for both models and establish corresponding lower bounds.
	This paper delivers the following messages:
	\begin{enumerate}
		\item In the noisy data setting, it is possible to approximate $F_p$ of the {\em unknown} ground-truth dataset using sublinear space in the data stream model and sublinear communication in the coordinator model, provided that the approximation guarantee is parameterized by a data-dependent quantity, referred to as the {\em $F_p$-mismatch-ambiguity}.
		
		\item In the data stream model, we establish tight space bounds for $F_p$ with respect to the input size in the noisy setting, showing that the $F_p$ problem is fundamentally more difficult in the presence of noise than in the noiseless case. In particular, whereas $F_2$ can be approximated in logarithmic space in the noiseless setting by the classical result of Alon, Matias, and Szegedy~\cite{AMS99}, in the noisy setting any algorithm for $F_2$ requires polynomial space.
		
		\item In the coordinator model, we show that, unlike in the noiseless setting, achieving polylogarithmic communication in the input size is generally impossible for $F_p$ in the presence of noise. Nevertheless, when the $F_p$-mismatch-ambiguity is below a certain threshold, communication that is {\em independent} of the input size becomes achievable.
	\end{enumerate}
	
	\vspace{2mm}
	\noindent{\bf A Novel Framework for Statistical Estimations on Noisy Data. \ }
	We start by introducing our new framework for analyzing noisy data.   Let $\sigma = (\sigma_1, \ldots, \sigma_m)$ be a noisy dataset of observed items. Each item $\sigma_i$ is generated from an underlying element $\tau_i$ by the addition of unknown noise. The ground-truth elements $\tau_i$ are drawn from a hidden universe $U = \{u_1, \ldots, u_n\}$. We impose no assumptions on the form or magnitude of the noise. As a result, the set of possible observed items may be infinite and strictly contains the ground-truth universe $U$.
	
	We assume access to an oracle that, for any pair of observed items $\sigma_i$ and $\sigma_j$, determines whether they are similar, denoted $\sigma_i \sim \sigma_j$. Each item is similar to itself, and elements in $U$ are pairwise dissimilar. Note that, because the noise is arbitrary, the similarity relation may produce both false positives and false negatives: two items may be declared similar even if they originate from different ground-truth elements, and two items may be declared dissimilar even if they originate from the same ground-truth element.
	
	We model a dataset $\sigma$ of size $m$ as a graph $G^\sigma = ([m], E^\sigma)$, with node $i$ representing $\sigma_i$. The edge set $E^\sigma$ is defined as $\{(i, j) \in [m]^2 \mid i \le j, \sigma_i \sim \sigma_j\}$.  Since we always have $\sigma_i \sim \sigma_i$, each node in $G^\sigma$ contains a self-loop.  For each $i \in [m]$, let $B_i^\sigma = \{j \in [m]\ |\ \sigma_j \sim \sigma_i\}$ be the set of nodes whose corresponding items are similar to $\sigma_i$ (including $\sigma_i$ itself), and let $d_i^\sigma = \abs{B_i^\sigma}$ be its size (or, the degree of the $i$-th node).
	
	It is easy to see that for a noiseless dataset $\tau$ of length $m$, $G^\tau = ([m], E^\tau)$ contains a set of disjoint cliques, each corresponding to a distinct universe element in $\tau$. We refer to such a graph $G^\tau$ as a {\em cluster graph}.\footnote{An equivalent and intuitive way to understand a cluster graph is that it is a ``P3-free'' graph. This means that there is no induced path of length three.}
	
	For a statistical function $f(\cdot)$ and a noisy dataset $\sigma = (\sigma_1, \ldots, \sigma_m)$, we define $f(\sigma)$ as $f(\tau)$, where $\tau = (\tau_1, \ldots, \tau_m)$ represents the ground truth of $\sigma$, and $\tau_i$ is the ground truth element of $\sigma_i$. Since $\tau$ is noiseless, $f(\tau)$ is always well defined. However, $\tau$ itself is unknown.  Our goal is to design an algorithm that approximates $f(\tau)$ using input dataset $\sigma$.
	
	At first glance, this task might seem impossible to achieve. However, as we will demonstrate in this paper using $F_p$ as an example, it is feasible to approximate $f(\tau)$ by introducing a parameter $\eta_f(\sigma, \tau)$ into the approximation ratio. This parameter, called the {\em $f$-mismatch-ambiguity}, captures how much $\sigma$ and $\tau$ differ with respect to the statistical function $f$. 
	\medskip
	
	We start by formally defining frequency moments for a noiseless dataset.
	\begin{definition}[Frequency Moments]
		Let $U=(u_1, \ldots, u_n)$ be a finite universe and $\tau=(\tau_1, \ldots, \tau_m) \in U^m$ be a noiseless dataset. For each element $u_j \in U$, define its frequency in $\tau$ as
		$$
			f_j = \abs{\{i \in [m] \mid \tau_i = u_j\}}.
		$$
		The $p$-th frequency moment of $\tau$ is defined as
		$
		F_p(\tau) = \sum_{j=1}^n f_j^p.
		$
	\end{definition}
	
	We next define a mismatch-ambiguity parameter for the $F_p$ problem. 
	\begin{definition}[$F_p$-mismatch-ambiguity ($p \ge 1$)]
		For a noisy dataset $\sigma = (\sigma_1, \ldots, \sigma_m)$ and its ground truth $\tau = (\tau_1, \ldots, \tau_m)$, the $F_p$-mismatch-ambiguity of $\sigma$ with respect to $\tau$ is defined as
		\begin{equation}
			\eta_p(\sigma, \tau) = \frac{1}{F_p(\tau)}\sum_{i \in [m]} \left(\abs{B_i^{\sigma} \cup B_i^{\tau}}^{p-1} - \abs{B_i^{\sigma} \cap B_i^{\tau}}^{p-1} \right).
		\end{equation}
		When there is no confusion, we simply write $\eta_p(\sigma, \tau)$ as $\eta_p$.
	\end{definition}
	We would like to make a few remarks regarding our choice of ambiguity definition for $F_p$.  
	\begin{itemize}
		\item  When $\sigma$ itself is noiseless, we have for all $i \in [m]$, $B_i^{\sigma} = B_i^{\tau}$, and consequently $\eta_p = 0$.   
		
		\item When $p = 1$, we always have $\eta_1 = 0$. Indeed, for the case when $p = 1$, we just need to count the number of items in the data stream, which is trivial even in the noisy setting.
		
		\item We call an edge present in $G^\sigma$ but absent in $G^\tau$ a {\em false positive} edge, and an edge present in $G^\tau$ but absent in $G^\sigma$ a {\em false negative} edge.  For $p = 2$, the term $\sum_{i \in [m]} \left(\abs{B_i^{\sigma} \cup B_i^{\tau}} - \abs{B_i^{\sigma} \cap B_i^{\tau}} \right)$ in $\eta_2$ is equal to twice of the sum of the false positive and false negative edges. We refer to the false positive and false negative edges as {\em mismatches}, which motivates the name “mismatch ambiguity” for this family of parameters.  
		
		\item We observe that the minimum number of mismatch edges in $G^\sigma$, taken over all possible ground truths $\tau$, coincides with the correlation clustering cost of $G^{cor} = ([m], E^{cor})$, where each pair $(i, j)$ satisfying $\sigma_i \sim \sigma_j$ contributes a `$+1$' edge to $E^{cor}$, and each pair $(i, j)$ with $\sigma_i \not\sim \sigma_j$ contributes a `$-1$' edge.\footnote{In the correlation clustering problem, we have a complete graph $G = (V, E)$, where each edge in $E$ is labeled either `$+1$' or `$-1$'. The objective is to partition $V$ into clusters so as to minimize the total number of misclustered edges. An edge is considered misclustered if it is labeled `$+1$' and its end nodes belong to different clusters, or if it is labeled `$-1$' and both end nodes are assigned to the same cluster.}
		
		\item We scale the sum $\sum_{i \in [m]} \left(\abs{B_i^{\sigma} \cup B_i^{\tau}}^{p-1} - \abs{B_i^{\sigma} \cap B_i^{\tau}}^{p-1} \right)$ by $\frac{1}{F_p(\tau)}$ to measure the impact of the mismatch edges on the quantity $F_p(\tau)$ we try to estimate.  The scaling factor is introduced for convenience; it transforms an additive error into a relative one for the ease of expressing the approximation ratio.
		
	\end{itemize}
	
	The appropriate definition of $f$-mismatch-ambiguity may vary across different statistical problems. In general, however, this parameter should be expressed as a function of the false positive and false negative edges in the graph representation of the input dataset with respect to the corresponding ground truth.
	
	In this paper, we are mainly interested in the case where the mismatch ambiguity is small, which we believe is common in practice given that we have a good similarity oracle.  On the other hand, when the dataset exhibits large mismatch ambiguity, standard statistical measures may lose their significance, and it may become infeasible to estimate $f(\tau)$---where $\tau$ denotes the ground truth---based solely on the noisy input $\sigma$.  
	
	\vspace{2mm}
	\noindent{\bf Preliminaries.\ }
	In this paper, we always assume that $p$ is a constant.  When we write $F_p$, we always mean $F_p(\tau)$, where $\tau$ is the unknown ground truth.  
	We often write $\eta_p(\sigma, \tau)$ as $\eta_p$ for brevity, as $\sigma$ is typically clear from context and $\tau$, the unknown ground truth, is always implicitly defined.
	
	We use $[n]$ to denote $\{1, \ldots, n\}$. All logarithms are taken to base $2$, unless stated otherwise. When we write $x = a \pm b$, we mean $x \in [a-b, a+b]$.  We say $\tilde{X}$ is an $(\eps, \delta)$-approximation of $X$ if with probability at least $(1 - \delta)$, we have $(1 - \eps)X \le \tilde{X} \le (1+\eps)X$.  We assume each item in the dataset fits one word.
	
	Note that in the noiseless setting, upper bounds are typically stated in terms of the universe size $n$ rather than the stream length $m$, under the standard assumption that $n$ and $m$ are polynomially related. In the noisy setting, however, the universe size is not explicitly defined, as it may be unbounded. Accordingly, we assume that the algorithm is given $m$, which can be regarded as an upper bound on the number of distinct items (and hence the effective universe size), since elements absent from the stream can be ignored. If $m$ is not known in advance, an additional pass (in the streaming model) or communication round (in the coordinator model) is required to learn $m$.
	
	To facilitate comparison between our bounds in the noisy setting and the existing results in the noiseless setting, we also use $m$ to denote the universe size in the noiseless setting. Notably, in the hard input constructions for the $F_p$ lower bound in the noiseless setting (e.g., \cite{BYJ+04,CKS03}), it consistently holds that $n = \Theta(m)$.

	\vspace{2mm}
	\noindent{\bf Our Results.\ }
	Our algorithmic result in the data stream model can be stated as follows:
	\begin{itemize}
		\item There is a one-pass $((\eps+O(\eta_p)), 0.01)$-approximation algorithm for $F_p\ (p \in \mathbb{Z}^+)$ over a data stream of $m$ items with mismatch ambiguity $\eta_p \le \frac{1}{3(p!)}$. The algorithm uses $O\left( \frac{1}{\eps^2}m^{1-{1}/{p}}\right)$ words of space. 
	\end{itemize} 
	We supplement the algorithmic result with the following lower bound:
	\begin{itemize}
		\item For any constant $C \ge 0$, any $O(1)$-pass $((\epsilon + C\eta_p), 0.48)$-approximation algorithm for $F_p\ (p > 1)$ over a data stream of $m$ items must use at least ${\Omega}\left( \frac{1}{\epsilon^{1/p}} m^{1 - 1/p} \right)$ bits of space.
	\end{itemize}
	This lower bound result contrasts with the noiseless setting, where for $F_p\ (p > 1)$, there exists one-pass $(\eps, 0.01)$-approximation algorithms (e.g., \cite{BO13}) for data streams containing $m$ items, using $\tilde{O}\left(\frac{1}{\eps^2}m^{1-{2}/{p}}\right)$ bits of space.\footnote{For simplicity, we hide polylogarithmic factors in the $\tilde{O}(\cdot)$ notation.}  In particular, the space complexity for $p = 2$ is logarithmic in the noiseless setting, but is polynomial in $m$ in the noisy setting.  
	
	In the coordinator model, we obtain the following algorithmic result:
	\begin{itemize}
		\item There is a two-round $((\eps+O(\eta_p)), 0.01)$-approximation algorithm for $F_p\ (p \ge 1)$ over a dataset of $m$ items with mismatch ambiguity $\eta_p \le 0.4$ distributed across $k$ sites. The algorithm uses $O\left( \frac{1}{\eps^2}km^{1-{1}/{p}}\right)$ words of communication. 
	\end{itemize} 
	For the lower bound, we show:
	\begin{itemize}
		\item For any constant $C \ge 0$, any $((\epsilon + C\eta_p), 0.48)$-approximation algorithm for $F_p\ (p > 1)$ over a dataset of $m$ items distributed across $k \ge 10p(C+1)$ sites must use at least ${\Omega}\left( \frac{1}{\epsilon^{1/p}} m^{1 - 1/p} \right)$ bits of communication.  
	\end{itemize}	
	This lower bound result is in stark contrast with the noiseless setting, where $(\eps, 0.01)$-approximation algorithms for $F_p$ exist using only ${O}\left(\frac{1}{\eps^2}k^{p-1} \log^{O(1)} m\right)$ bits of communication~\cite{HXZ+25};  in particular, the dependency on the input size $m$ is polylogarithmic. 
	
	For both the data stream model and the coordinator model, our upper and lower bounds match in terms of $m$, differing only in whether the cost is measured in words or in bits.
	
	Finally, we show that in the coordinator model, when the mismatch ambiguity is sufficiently small, we can design an algorithm that is significantly more efficient, achieving communication cost close to that in the noiseless setting in terms of $k$ and {\em independent} of the dataset size $m$.
	\begin{itemize}
		\item  There is a three-round $((\eps+O(\eta_p)), 0.01)$-approximation algorithm for $F_p\ (p \ge 2)$ over a dataset of $m$ items with mismatch ambiguity $\eta_p \le \frac{\eps^p}{4^{p+1} \cdot k^{p-1}}$ distributed across $k$ sites. The algorithm uses   $O\left(\frac{1}{\eps^{p+1}}k^p\right)$ words of communication.
	\end{itemize}	
	For $p = 2$, the mismatch ambiguity threshold in the above algorithm is in the order of ${\eps^2}/{k}$, and the algorithm achieves a communication cost of $O\left(k^2/\eps^3\right)$. While in our lower bound proof for the coordinator model, the hard input has a mismatch ambiguity in the order of ${\epsilon}/{k}$, which yields an $\Omega\left(\sqrt{{m}/{\eps}}\right)$ communication lower bound. Taken together, these results reveal a rough phase transition in the communication cost as a function of the mismatch ambiguity in the coordinator model.

	\subsection{Related Work}
	\label{sec:related}
	Several research directions in the literature are closely related to our work.
	
	\vspace{2mm}
	\noindent{\bf Frequency Moments on Noiseless Data.\ }
	The $F_p$ problem has been extensively studied in the noiseless setting, both in the data stream model~\cite{LW13,AMS99,IW05,BGK+06,MW10,AKO11,BO10,Andoni17, Ganguly11, Woodruff04, BYJ+04, CKS03, Ganguly12, BZ25, GW18} and the coordinator model~\cite{HXZ+25,CMY08, WZ12, EKM+24, KVW14, JW23}.  In the data stream model, for $p > 2$, the best known space upper bound is $O\left(\frac{1}{\eps^2} n^{1-2/p}\log^2 n\right)$ bits~\cite{Ganguly11}, where $n$ denotes the universe size and the stream length is assumed to satisfy $m = n^{\Theta(1)}$.  This bound is tight up to a $\log n$ factor~\cite{LW13}.  In the coordinator model, for $p \ge 2$, there is a two-round algorithm that achieves a communication cost of $O\left(\frac{1}{\eps^2}k^{p-1} \log^{O(1)} n\right)$ bits~\cite{HXZ+25}, which is tight up to a $\log^{O(1)} n$ factor~\cite{WZ12}. 
	
	\vspace{2mm}
	\noindent{\bf Statistical Estimations on Noisy Data.\ }  
	The literature on handling noisy datasets in the data stream model is rather sparse. The authors of \cite{CZ16} and \cite{CZ18} investigated the distinct elements problem ($F_0$) and $\ell_0$-sampling, but their results are restricted to constant-dimensional Euclidean spaces and metric spaces that support efficient locality-sensitive hashing. A recent work \cite{Zhang25} broadens the scope to general metric spaces (or more generally, any domain equipped with a 0/1 similarity function). However, the data ambiguity model used in all these works is tailored specifically to the $F_0$ problem---more precisely, it is defined as the difference between the minimum-cardinality disjoint clique partition of the graph $G^\sigma$ and the closest cluster graph---and does not extend naturally to other statistical problems.
	
	In the coordinator model, the only known work on statistical estimation on noisy data is \cite{Zhang15}. However, it focuses exclusively on datasets whose graph representations form cluster graphs.
	
	\vspace{2mm}
	\noindent{\bf Streaming Algorithms on Uncertain Data.\ }
	A line of work has studied streaming algorithms for uncertain/probabilistic 
	data~\cite{JMM+07,JKV07,CG07,JYC+08,ZLY08}, where data noise is captured 
	through the ``possible-worlds'' semantics. In this framework, each element is 
	associated with a probability distribution over its possible occurrences, 
	making it suited for noise sources such as misspellings. However, 
	such method is only feasible when the universe is fixed and known in advance, which often does not hold in real-world applications. For example, if a data item is an image perturbed from an original copy that is never observed, then it is impossible to design an oracle that always returns the correct ground-truth element. As another example, two queries to a large language model may yield answers with nearly identical semantics but slightly different surface forms. In this case, the notion of a "ground-truth element" is conceptual and difficult to make explicit, and the number of possible semantics can be unbounded. Therefore, our similarity-oracle based framework is suitable for a much broader class of noisy data.

	\section{The Algorithms}
	\label{sec:algo}
	
	In this section, we present algorithms for $F_p$ on noisy datasets in the data stream and coordinator models.
	
	We use $\MC(i_1, \dots, i_p)\ (i_1, \ldots, i_p \in [m])$ to denote the multinomial coefficient, defined as:
	$
	\MC(i_1, \dots, i_p) = \frac{p!}{\prod_{k \in [m]} \left(\sum_{j \in [p]} \mathbf{1}\{i_j = k\}\right)!}\ .
	$
	In words, $\MC(i_1, \dots, i_p)$ is the number of permutations that can be formed by $\{\sigma_{i_1}, \ldots, \sigma_{i_p}\}$.

	\subsection{The Streaming Algorithm}
	\label{sec:stream-algo}
	
	We begin with a one-pass streaming algorithm for $F_p$ with an approximation guarantee parameterized by $\eta_p$.  
	
	\vspace{2mm}
	\noindent{\bf Ideas and Technical Overview.\ }
	We say a tuple of nodes $(i_1, \ldots, i_p)$ is a $p$-clique in $G^\sigma$ if, for all $j, k \in [p]$, $\sigma_{i_j} \sim \sigma_{i_k}$, and a $p$-clique  is {\em ordered} if the sequence of its $p$ nodes matters.  As nodes in our setting are considered similar to themselves, we allow cliques to include duplicate nodes.
	
	We note that for the ground truth data stream $\tau$, $F_p$ is precisely $\abs{\K_p^\tau}$, where $\K_p^\tau$ is the set of ordered $p$-cliques in $G^\tau$.\footnote{For example, for $p = 3$ and $G^{\tau} = (V, E)$ where $V = \{1, 2\}$ and $E = \{(1, 1), (1, 2), (2, 2)\}$, then $\K_p^\tau$ contains $(1, 1, 1), (1, 1, 2), (1, 2, 1), (2, 1, 1), (1, 2, 2), (2, 1, 2), (2, 2, 1)$ and $(2, 2, 2)$.}   For a noisy data stream $\sigma$, let $\K_p^\sigma$ be the set of ordered $p$-cliques in $G^\sigma$. Our main observation is that for {\em any} ground truth $\tau$ such that $\eta_p$ is small, we have $\abs{\K_p^\sigma} \approx  \abs{\K_p^\tau} = F_p$. Therefore, if we can accurately estimate $\abs{\K_p^\sigma}$, we can also well-estimate $F_p$.
	
	The question now becomes how to estimate $\abs{\K_p^\sigma}$ in one pass.  A natural idea is to draw uniform random $p$-tuples from $[m]^p$ and then count how many of them belong to $\K_p^\sigma$. However, this approach has a problem: $\abs{\K_p^\sigma}$ may be as small as $m$, resulting in a low probability that a sample belongs to $\K_p^\sigma$. On the other hand, directly counting $\abs{\K_p^\sigma}$ is challenging in the data stream model, since the nodes within a tuple are ordered, and this ordering typically does not align with the order in which the nodes appear in the stream. We thus focus on {\em increasingly ordered} $p$-cliques, where the order of nodes in the tuple matches their order in the data stream.
	
	We design a sampling-and-counting procedure, denoted by $\A$, that leverages the observation that each increasingly ordered $j$-clique can be incrementally constructed from an increasingly ordered $(j-1)$-clique. Conceptually, the generative process can be represented as a tree, where the root corresponds to a $0$-clique and its children are $1$-cliques by adding one more node to the $0$-clique, continuing in this manner for $p$ levels. The leaves of the tree precisely enumerate the increasingly ordered $p$-cliques. We then perform a random walk from the root to a leaf to sample a path, and produces an estimate $X$ based on the probability of selecting this path and an appropriate multinomial coefficient. We can show that $X$ satisfies $\E[X] = |\mathcal{K}_p^\sigma|$.
	
	The critical step in the above approach is to determine the number of unbiased estimates needed for the final approximation.   Let $X_1, \ldots, X_t$ be estimates returned by $t$ independent runs of $\A$, where for every $j \in [t]$, $\E[X_j] = \abs{\K_p^\sigma}$.  Let $\bar{X} = \frac{1}{t} \sum_{j \in [t]} X_j$. If $\sigma$ is noiseless, it is not difficult to show that $t = O\left(\frac{1}{\eps^2}m^{1-1/p}\right)$ estimates are sufficient to ensure that $\bar{X}$ tightly concentrates around its expectation.  However, in the noisy setting, $\bar{X}$ may exhibit high instability in worst case scenarios.  For example, consider the case when $p = 2$ and $G^\sigma$ is a star graph. In this case, $\K_2^\sigma$ consists of all ordered pairs of similar  nodes, with approximately two-thirds of the pairs containing the central node. If the central node arrives first in the stream, then with probability $\left(1 - \frac{1}{m}\right)$ it is not sampled. As a result, roughly two-thirds of the pairs in $\K_2^\sigma$ are not counted, which prevent us to obtain a $1.01$-approximation with probability $0.99$ when $t = o(m)$. 
	
	What we manage to show is that when $\eta_p$ is small, we can estimate $\abs{\K_p^\sigma}$ accurately using a sublinear number of samples.  First, it is not difficult to establish an upper bound of $\bar{X}$ by a Markov inequality.  In fact, as $\bar{X}$ is the sum of independent random variables, we can apply an inequality from \cite{Feige04} to achieve a better bound.  The challenge lies in establishing the lower bound of $\bar{X}$.  The standard method is to bound the variance of $\bar{X}$ (or, the variance of each $X_j\ (j \in [t])$) and then apply Chebyshev's inequality. However, if we analyze the variance of $X_j$ directly without taking into account the mismatch ambiguity $\eta_p$, the variance could be very high.  Consider again the star graph, which has a high $\eta_p$ with respect to any ground truth.  Easy calculation gives $\E[X_j] = \Theta(m)$ and $\Var[X_j] = \Omega(m^p)$. Consequently, we have to use $t = \Omega\left(m^{p-2}\right)$ copies of $X_j$ to bring down the variance of $\bar{X}$ to $O(1)$.  Therefore, we must account for $\eta_p$ in our analysis. The difficulty is that the construction of $\bar{X}$ cannot incorporate $\eta_p$, as it is an unknown quantity to the algorithm.
	
	We circumvent the above issue by relating each random variable $X_j\ (j \in [t])$ to another random variable 
	$Y_j \triangleq \min\left\{X_j, m (p!) \left(d_{I_1}^\tau\right)^{p-1}\right\},$
	where $I_1$ is distributed uniformly at random in $\{1, \ldots, m\}$.  The intuition of creating such a random variable $Y_j$ is the following: recall that if we apply $\A$ to the ground truth $\tau$, whose graph representation $G^\tau$ is a cluster graph, we can obtain a tighter upper bound on the variance of $X_j$, and consequently only need $t = o(m)$ runs of $\A$.  Now, if $\eta_p$ is small, we would expect that  $\K_p^\sigma$ and $\K_p^\tau$ overlap significantly, and consequently, $\abs{\K_p^\sigma}$ and $\abs{\K_p^\sigma \cap \K_p^\tau}$ are close.  We can intuitively interpret $Y_j$ with $\E[Y_j] \approx \abs{\K_p^\sigma \cap \K_p^\tau}$ as being close to $X_j$ with $\E[X_j] = \abs{\K_p^\sigma}$ but with a smaller variance.  Therefore, we can first establish a lower bound for $Y_j$, which also provides a lower bound for  $X_j$.

	\vspace{2mm}
	\noindent{\bf Algorithm and Analysis.\ }
	The algorithm is described in   Algorithm~\ref{alg:Fp-one-pass}, which uses Algorithm~\ref{alg:Fp-one-sample} as a subroutine.  At a high level, Algorithm~\ref{alg:Fp-one-sample} aims to obtain an unbiased estimator of the number of ordered $p$-cliques in $G^\sigma$, which serves as a low-bias estimate of $F_p$. This is achieved by sampling and counting increasingly ordered $p$-cliques in its main loop (Lines~\ref{ln:start-loop}--\ref{ln:end-loop}), followed by scaling the result with appropriate multinomial coefficients (Line~\ref{ln:scale}). By executing Algorithm~\ref{alg:Fp-one-sample} multiple times in parallel, Algorithm~\ref{alg:Fp-one-pass} applies a min-average aggregation (Lines~\ref{ln:min-avg-1}--\ref{ln:min-avg-2}) to produce the final approximation for $F_p$.  
	
	We have the following result regarding Algorithm~\ref{alg:Fp-one-pass}. 
	
	\begin{algorithm}[t]
		\caption{Streaming-Robust-$F_p$-One-Sample}
		\label{alg:Fp-one-sample}
		\DontPrintSemicolon
		\SetAlgoNoEnd
		
		\KwIn{a noisy data stream $\sigma$}
		\KwOut{an estimate of $F_p$}
		
		$r_1, \dots, r_p,  i_1, \dots, i_p \gets 0$, $a_1, \dots, a_p \gets \Null$\;
		\ForEach{incoming $\sigma_j$\label{ln:start-loop}} {
			$r_1 \gets r_1 + 1$\;
			\text{with probability} $\frac{1}{r_1} $, $(a_1, i_1, r_2) \gets (\sigma_j, j, 0)$\;
			\If{$a_1 \sim \sigma_j$} {
				$r_2 \gets r_2 + 1$\;
				\text{with probability} $\frac{1}{r_2}$, $(a_2, i_2, r_3) \gets (\sigma_j, j, 0)$\;
				\If{$a_2 \sim \sigma_j$} {
					$r_3 \gets r_3 + 1$\;
					\text{with probability} $\frac{1}{r_3}$, $(a_3, i_3, r_4) \gets (\sigma_j, j, 0)$\;
					$\dots$\;
					\If{$a_{p-1} \sim \sigma_j$} {
						$r_p \gets r_p + 1$\;
						\text{with probability} $\frac{1}{r_p}$, $(a_p,i_p) \gets (\sigma_j, j)$\label{ln:end-loop}
					}
				}
			}
		}
		\Return $\MC(i_1, \dots, i_p) \prod_{k \in [p]} r_k$. \label{ln:scale}
	\end{algorithm}
	
	\begin{algorithm}[t]
		\caption{Streaming-Robust-$F_p$}
		\label{alg:Fp-one-pass}
		\DontPrintSemicolon
		\SetAlgoNoEnd
		
		\KwIn{a noisy data stream $\sigma$ of size $m$, parameter $\eps$}
		\KwOut{an $((\eps + O(\eta_p)), 0.01)$-approximation of $F_p$}
		
		$\ell \gets 100$, $t \gets \frac{10^6(p!)m^{1-{1}/{p}}}{\eps^2}$\; \label{ln:t}
		
		obtain $\ell t$ estimates $\{X_{i,j}\}_{i \in [\ell], j \in [t]}$ by running $\ell t$ instances of Algorithm~\ref{alg:Fp-one-sample} in parallel
		
		for $i \in [\ell]$, let $\Bar{X}_i \gets \frac{1}{t} \sum_{j \in [t]} X_{i,j}$\;  \label{ln:min-avg-1}
		\Return $\min\{\Bar{X}_1, \dots, \Bar{X}_{\ell}\}$.  \label{ln:min-avg-2}
	\end{algorithm}

	\begin{theorem}
		\label{thm:Fp-one-pass}
		For any constant $p \in \mathbb{Z}^+$, given a noisy input data stream of length $m$ with $\eta_p \leq \frac{1}{3(p!)}$, Algorithm~\ref{alg:Fp-one-pass} computes an $((\eps+O(\eta_p), 0.01)$-approximation of  $F_p$, using a single pass and $O\left(\frac{1}{\eps^2}m^{1-1/p}\right)$ words of space.
	\end{theorem}

	In the rest of this section, we prove Theorem~\ref{thm:Fp-one-pass}.
	We will make use of the following inequality.
	\begin{lemma}[\cite{Feige04}]
		\label{lem:feige}
		Let $Z_1, \dots, Z_t$ be arbitrary non-negative independent random variables, with expectations $\mu_1, \dots, \mu_t$ where $\mu_i \leq 1$ for all $i$. Then for any $\delta > 0$,
		\begin{equation*}
			\Pr\left[\sum_{i=1}^t Z_i < \left(\sum_{i=1}^t \mu_i\right) + \delta \right] > \min \left( \frac{\delta}{1 + \delta}, \frac{1}{13} \right).
		\end{equation*}
	\end{lemma}
	
	We note that when $p = 1$, the output of Algorithm~\ref{alg:Fp-one-pass} is always $m$, which is exactly $F_1$.  We thus focus on integers $p \ge 2$.  
	
	We start by introducing the following concept.
	\begin{definition}[Ordered $p$-Clique]
		\label{def:p-clique}
		For a data stream $\sigma$ of length $m$, define the set of ordered $p$-cliques in $G^\sigma$ as
		$
		\K_p^\sigma = \left\{(i_1, \ldots, i_p) \in [m]^{p} \mid \forall (j, k) \in [p]^2, \sigma_{i_j} \sim \sigma_{i_k} \right\}.
		$
	\end{definition}
	
	\begin{fact}
		\label{fact:Fp}
		For a noiseless data stream $\tau$, $\abs{\K_p^\tau} = F_p$.
	\end{fact}
	
	\begin{proof}
		Let $n_\tau$ be the number of distinct elements in $\tau$, and $\{V_1, \ldots, V_{n_\tau}\}$ be the set of cliques in $G^\tau$.  It is easy to see that for any $K_p \in \K_p^\tau$, all $p$ nodes in $K_p$ must belong to the same clique $V_k$ for some $k \in [n_\tau]$. Therefore, the total number of ordered $p$-cliques is $\sum_{k \in [n_\tau]} \abs{V_k}^p = F_p$.
	\end{proof}
	
	The following lemma shows that for a noisy data stream $\sigma$ and the corresponding ground truth  $\tau$, if $\eta_p$ is small, then $\K_p^\sigma$ and $\K_p^\tau$ overlap significantly. 
	
	\begin{lemma} 
		\label{lem:Fp-diff}
		For graphs $G^\sigma = ([m], E^\sigma)$ and $G^\tau = ([m], E^\tau)$, representing a noisy data stream $\sigma$ and the corresponding ground truth $\tau$ respectively, we have
		\begin{equation*}
			\abs{\K_p^{\sigma} \setminus \K_p^{\tau}} \leq \eta_p F_p, \quad \text{and} \quad \abs{\K_p^{\tau} \setminus \K_p^{\sigma}} \leq \frac{p}{2} \cdot \eta_p  F_p.
		\end{equation*}
	\end{lemma}
	
	\begin{proof}
		%	Abbreviate $\eta_p(\sigma, \tau)$ as $\eta_p$.
		We first bound the size of $\K_p^{\sigma} \setminus \K_p^{\tau}$. Consider an ordered $p$-clique $K_p = ({i_1}, \dots, {i_p})$ in $\K_p^{\sigma} \setminus \K_p^{\tau}$ with $i_1 = i$. There must exist a $k \in [p]$ such that ${i_k} \in B_i^{\sigma} \setminus B_i^{\tau}$. Therefore, the total number of ordered $p$-cliques  in $\K_p^{\sigma} \setminus \K_p^{\tau}$  with $i_1 = i$ is at most
		$
		\abs{B_i^{\sigma}}^{p-1} - \abs{B_i^{\sigma} \cap B_i^{\tau}}^{p-1} \leq \abs{B_i^{\sigma} \cup B_i^{\tau}}^{p-1} - \abs{B_i^{\sigma} \cap B_i^{\tau}}^{p-1}.
		$
		Summing over $i \in [m]$, we have 
		\begin{equation*}
			\abs{\K_p^{\sigma} \setminus \K_p^{\tau}} \le \sum_{i \in [m]} \left(\abs{B_i^{\sigma} \cup B_i^{\tau}}^{p-1} - \abs{B_i^{\sigma} \cap B_i^{\tau}}^{p-1}\right) = \eta_p  F_p.
		\end{equation*}
		
		We next bound the size of $\K_p^{\tau} \setminus \K_p^{\sigma}$. Let $n_i$ be the number of ordered $p$-cliques in $\K_p^\tau$ that contains $i$ and at least one node from $B_i^\tau \backslash B_i^\sigma$.  We have
		$$
		n_i \leq p  \left(\abs{B_i^{\tau}}^{p-1} - \abs{B_i^{\sigma} \cap B_i^{\tau}}^{p-1}\right) \leq p  \left( \abs{B_i^{\sigma} \cup B_i^{\tau}}^{p-1} - \abs{B_i^{\sigma} \cap B_i^{\tau}}^{p-1} \right).
		$$
		Note that for each $K_p =({i_1}, \dots, {i_p}) \in \K_p^{\tau} \setminus \K_p^{\sigma}$, there exists $x, y\in [p]$ such that ${i_x} \neq {i_y}$ and $({i_x}, {i_y}) \notin E^{\sigma}$, which implies that $K_p$ is counted in both $n_{i_x}$ and $n_{i_y}$. Summing over $i \in [m]$, we have
		\begin{equation*}
			\abs{\K_p^{\tau} \setminus \K_p^{\sigma}} \leq \frac{1}{2} \sum_{i \in [m]} n_i \leq \frac{p}{2}  \cdot \sum_{i \in [m]}\left( \abs{B_i^{\sigma} \cup B_i^{\tau}}^{p-1} - \abs{B_i^{\sigma} \cap B_i^{\tau}}^{p-1} \right) = \frac{p}{2} \cdot \eta_p  F_p.
		\end{equation*}
	\end{proof}
	
	Fact~\ref{fact:Fp} and Lemma~\ref{lem:Fp-diff} imply that we can estimate $F_p$ via estimating $\abs{\K_p^\sigma}$.  However, it seems difficult to estimate this quantity directly in the data stream model under sublinear memory space.  We therefore turn our attention to a closely related set.
	\begin{definition}[Increasingly Ordered $p$-Clique]
		For a data stream $\sigma$, define the set of increasingly ordered $p$-cliques in $G^\sigma$ as
		$\vec{\K}_p^\sigma = \{(i_1, \ldots, i_p) \in \K_p^\sigma \mid i_1 \le \ldots \le i_p\}$.
	\end{definition}
	
	Note that for each increasingly ordered $p$-clique $\vec{K}_p = (i_1, \dots, i_p) \in \vec{\K}_p^\sigma$, there are exactly $\MC(i_1, \dots, i_p)$ ordered $p$-cliques in $\K_p^{\sigma}$ that consist of the same multiset of nodes. 
	
	To facilitate the subsequent analysis, we introduce the following concept.
	\begin{definition} [Tail Set]
		For a data stream $\sigma$ and $j \in [p]$, define the tail set of an increasingly ordered $j$-clique $(i_1, \dots, i_j)$ in $G^\sigma$ as
		$
		\Tail(i_1, \dots, i_j) = \{i \in [m] \mid  i \geq i_j \land (\forall k \in [j], \sigma_i \sim \sigma_{i_k})\}.
		$
	\end{definition}
	In words, $	\Tail(i_1, \dots, i_j)$ contains all nodes $i \in [m]$ appearing after $i_j$ (including $i_j$) in the stream that, together with the $j$-clique $(i_1, \dots, i_j)$, form an increasingly ordered $(j+1)$-clique.
	
	The following lemma shows that Algorithm~\ref{alg:Fp-one-sample} outputs an unbiased estimator of $\abs{\K_p^\sigma}$.  Let $I_1, \dots, I_p$ correspond to $i_1, \ldots, i_p$  in Algorithm~\ref{alg:Fp-one-sample} be the random nodes, and let $R_1, \dots, R_p$ correspond to $r_1, \ldots, r_p$ in Algorithm~\ref{alg:Fp-one-sample} be the tail counts. Note that $R_1 = m$ is actually a fixed value, whereas others are true random variables.
	
	\begin{lemma} 
		\label{lem:Fp-exp}
		Let $X = \MC(I_1, \dots, I_p) \prod_{k \in [p]} R_k$ be the output of Algorithm~\ref{alg:Fp-one-sample}. It holds that $\E[X] = \abs{\K_p^\sigma}$.
	\end{lemma}
	
	\begin{proof}
		It is easy to see that in Algorithm~\ref{alg:Fp-one-sample}, $I_1$ is uniformly sampled from $[m]$ and $I_k$ is uniformly sampled from $\Tail(I_1, \dots, I_{k-1})$ for $k = 2, \dots, p$. Let $\vec{K}_p = (i_1, \dots, i_p) \in \vec{\K}_p^{\sigma}$ be an increasingly ordered $p$-clique. We have
		\begin{eqnarray}
			\Pr[I_1 = i_1, \dots, I_p = i_p] &=& \Pr[I_1 = i_1] \cdot \prod_{k=2}^{p}\Pr\left[I_k = i_k \mid I_1 = i_1, \dots, I_{k-1} = i_{k-1}\right] \nonumber\\
			&=& \frac{1}{m}  \prod_{k=2}^{p} \frac{1}{\abs{\Tail(i_1, \dots, i_{k-1})}}. \label{eq:c-1}
		\end{eqnarray}
		
		Note that $R_1 = m$, and $R_k = \abs{\Tail(I_1, \dots, I_{k-1})}$ for $k = 2, \dots, p$, and $(I_1, \dots, I_p)$ is always an increasingly ordered $p$-clique. It follows that
		\begin{eqnarray*}
			\E[X] &=& \sum_{(i_1, \dots, i_p) \in \vec{\K}_p^{\sigma}} \Pr[I_1 = i_1, \dots, I_p = i_p] \cdot \E[X \mid I_1 = i_1, \dots, I_p = i_p] \\
			&\stackrel{\eqref{eq:c-1}}{=}& \sum_{(i_1, \dots, i_p) \in \vec{\K}_p^{\sigma}}  \frac{1}{m} \prod_{k=2}^{p} \frac{1}{\abs{\Tail(i_1, \dots, i_{k-1})}} \cdot  \MC(i_1, \dots, i_p) \cdot m \prod_{k=2}^{p} \abs{\Tail(i_1, \dots, i_{k-1})} \\
			&=&\sum_{(i_1, \dots, i_p) \in \vec{\K}_p^{\sigma}} \MC(i_1, \dots, i_p) = \abs{\K_p^{\sigma}}.
		\end{eqnarray*}
	\end{proof}
	
	Let $X_1, \dots, X_t$ denote $t$ independent outputs of Algorithm~\ref{alg:Fp-one-sample}, and let $\Bar{X} = \frac{1}{t} \sum_{i \in [t]} X_i$. The next two lemmas show that when the mismatch ambiguity $\eta_p$ is small, $\Bar{X}$ is tightly concentrated around $F_p$ with good probability.
	
	\begin{lemma}[$\Bar{X}$ is not too large] \label{lem:Fp-large}
		$
		\Pr\left[\Bar{X} < (1 + \eps + \eta_p) \cdot F_p \right] > \frac{1}{13}.
		$
	\end{lemma}
	
	\begin{proof}
		Lemma~\ref{lem:Fp-exp} gives $\E[\Bar{X}] = \abs{\K_p^{\sigma}}$ . Applying Lemma~\ref{lem:feige} (setting $Z_i = X_i/|\K_p^{\sigma}|$ and $\delta = 1$), we have
		\begin{equation}
			\Pr \left[\Bar{X} < \left(1 + \frac{1}{t}\right) \abs{\K_p^{\sigma}} \right] > \frac{1}{13}. \label{eq:apply-feige}
		\end{equation}
		By Fact~\ref{fact:Fp} and  Lemma~\ref{lem:Fp-diff},
		\begin{equation}
			\abs{\K_p^{\sigma}} \leq \abs{\K_p^{\tau}} + \abs{\K_p^{\sigma} \setminus \K_p^{\tau}} \leq (1 + \eta_p)  F_p. \label{eq:high0}
		\end{equation}
		The lemma follows from \eqref{eq:apply-feige}, \eqref{eq:high0}, and the fact that $\frac{1}{t} \ll \frac{\eps}{2}$.
	\end{proof}
	
	\begin{lemma} [$\Bar{X}$ is not too small] \label{lem:Fp-small}
		$
		\Pr\left[\Bar{X} \ge \left( 1 - \eps - 2(p!)\eta_p \right) \cdot  F_p\right] > 1 -  2\cdot10^{-5}.
		$
	\end{lemma}
	
	\begin{proof}
		%	Let $X$ be a random variable distributed the same as $X_1, \ldots, X_t$.  
		For every $j \in [t]$, let 
		$Y_j = \min\left\{X_j, m (p!) \abs{B_{I_1}^{\tau}}^{p-1}\right\}.$
		We have the following claims, whose proofs will be given shortly. The first claim states that the expectation of each $Y_j$ is not far away from $F_p$ given that $\eta_p$ is small, and the second bounds the variance of each $Y_j$.
		\begin{claim} \label{claim:low1}  For any $j \in [t]$,
			$
			(1 - 2(p!) \eta_p)  F_p \leq \E[Y_j] \leq (1 + \eta_p)  F_p.
			$
		\end{claim}
		\begin{claim} \label{claim:low2}  For any $j \in [t]$,
			$
			\frac{\Var[Y_j]}{(\E[Y_j])^2} \leq \frac{p! }{1 - 2(p!) \eta_p} \cdot m^{1-1/p}.
			$
		\end{claim}
		
		We now have
		\begin{eqnarray}
			\Pr\left[\Bar{X} < \left( 1  - \eps - 2(p!)\eta_p \right) F_p\right] &=& \Pr\left[ \frac{1}{t} \sum_{j \in [t]} X_j - (1 - 2(p!)\eta_p)F_p < - \eps F_p \right] \label{eq:low1} \\
			&\leq& \Pr\left[ \frac{1}{t} \sum_{j \in [t]}Y_j - \E[Y_j] < - \eps F_p \right] \label{eq:low2} \\
			&\leq& \frac{\Var[Y_j]}{t \eps^2 F_p^2}  \label{eq:low21} \label{eq:low21}\\
			& \leq& \frac{(1 + \eta_p)^2 \Var[Y_j]}{t \eps^2 (\E[Y_j])^2} \label{eq:low3} \\
			&\leq& \frac{(1 + \eta_p)^2 }{t \eps^2 } \cdot \frac{p! \cdot m^{1-1/p}}{1 - 2(p!) \eta_p} \label{eq:low4} \\
			&\leq& \frac{1}{t} \cdot \frac{12 (p!) m^{1-1/p}}{\eps^2} \label{eq:low5} \\
			&\leq& 2\cdot 10^{-5}, \nonumber
		\end{eqnarray}
		where from \eqref{eq:low1} to \eqref{eq:low2}, we use the fact $X_j \geq Y_j$ and the first inequality of Claim~\ref{claim:low1}. From  \eqref{eq:low21} to \eqref{eq:low3}, we apply the second inequality of Claim~\ref{claim:low1}.  From  \eqref{eq:low3} to \eqref{eq:low4}, we apply Claim~\ref{claim:low2}. From \eqref{eq:low4} to \eqref{eq:low5}, we use our assumption $\eta_p < \frac{1}{3(p!)}$. The last inequality is due to our choice of $t$ (Line~\ref{ln:t} of Algorithm~\ref{alg:Fp-one-pass}).
	\end{proof}
	
	Now we are ready to prove Theorem~\ref{thm:Fp-one-pass}. Recall that we have set $\ell = 100$ in Algorithm~\ref{alg:Fp-one-pass}. Let $X_{\min} = \min\left\{\Bar{X}_1, \dots, \Bar{X}_{\ell}\right\}$ be the output of Algorithm~\ref{alg:Fp-one-pass}. By Lemma \ref{lem:Fp-large} and \ref{lem:Fp-small}, we have
	\begin{eqnarray*}
		\Pr \left[ X_{\min} < (1 + \eta_p + \eps) F_p \right] &\ge& 1 - \left( 1 - \frac{1}{13}\right)^{\ell} > 0.999, \quad \text{and}  \\
		\Pr \left[ X_{\min} \ge (1 - \eps - 2(p!)\eta_p) F_p \right] &\ge& 1 - \ell \cdot 2 \cdot 10^{-5} = 0.998.
	\end{eqnarray*}
	The theorem follows from a union bound.
	\medskip
	
	Finally, we give the proofs for the two claims used in Lemma~\ref{lem:Fp-small}.
	\begin{proof}[Proof of Claim~\ref{claim:low1}]
		We prove for every $j \in [t]$, and thus drop the subscript $j$ in $X_j$ and $Y_j$ and write them as $X$ and $Y$ for convenience. 
		We first look at the difference between $X$ and $Y$,
		\begin{equation*}
			\E[X - Y] = \sum_{i \in [m]} \Pr[I_1 = i] \cdot \E[X-Y \mid I_1 = i] = \frac{1}{m} \sum_{i \in [m]} \E[X-Y \mid I_1 = i].
		\end{equation*}
		Conditioned on $I_1 = i$, we have
		\begin{eqnarray}
			X &=& \MC(i, I_2, \dots, I_p) \cdot \prod_{k \in [p]} R_k \leq p! \cdot \prod_{k \in [p]} R_k  \nonumber \\
			&=& p! \cdot m \cdot \prod_{k=2}^p \abs{\Tail(i, \dots, I_{k-1})} \label{eq:low1-2} \\
			&\leq& p! \cdot m \cdot \abs{\Tail(i)}^{p-1}  \label{eq:low1-3} \\ 
			&\leq& m(p!) \abs{B_i^{\sigma}}^{p-1}, \label{eq:low1-4}
		\end{eqnarray}
		where from \eqref{eq:low1-2} to \eqref{eq:low1-3}, we use the fact that $\Tail(i, I_2, \dots, I_{p-1}) \subseteq \Tail(i, I_2, \dots, I_{p-2}) \subseteq \cdots \subseteq \Tail(i)$. The step from \eqref{eq:low1-3} to \eqref{eq:low1-4} follows from $\Tail(i) \subseteq B_i^{\sigma}$.
		
		We thus have
		\begin{equation*}
			\E[X - Y \mid I_1 = i] = \E\left[\left. \max \left\{X - m(p!)\abs{B_i^{\tau}}^{p-1}, 0\right\} \ \right|\ I_1 = i \right] \le m(p!) \cdot \max\left\{\abs{B_i^{\sigma}}^{p-1} - \abs{B_i^{\tau}}^{p-1}, 0 \right\}.
		\end{equation*}
		Noting that $\Pr[I_1 = i] = \frac{1}{m}$, we have
		\begin{eqnarray}
			\E[X - Y] &\leq& p! \cdot \sum_{i \in [m]} \max\left\{\abs{B_i^{\sigma}}^{p-1} - \abs{B_i^{\tau}}^{p-1}, 0 \right\} \nonumber \\
			&\leq& p! \cdot \sum_{i \in [m]} \left(\abs{B_i^{\sigma} \cup B_i^{\tau}}^{p-1} - \abs{B_i^{\tau} \cap B_i^{\sigma}}^{p-1} \right) \nonumber \\
			&=& p! \cdot \eta_p F_p. \label{eq:low1-5}
		\end{eqnarray}
		Consequently,
		\begin{eqnarray*}
			\E[Y] &=& \E[X] - \E[X-Y] \\
			&\ge& \abs{\K_p^{\sigma}} - p! \cdot \eta_p F_p \quad (\text{by \eqref{eq:low1-5} and Lemma~\ref{lem:Fp-exp}})\\
			&\geq& \abs{\K_p^{\tau}} - \abs{\K_p^{\tau} \setminus \K_p^{\sigma}} - p! \cdot \eta_p F_p \nonumber \\
			&\geq& (1 - 2(p!) \eta_p) F_p. \quad (\text{by Fact~\ref{fact:Fp} and Lemma~\ref{lem:Fp-diff}})
		\end{eqnarray*}
		On the other hand, by \eqref{eq:high0} and Lemma~\ref{lem:Fp-exp}, we have
		$
		\E[Y] \leq \E[X] = \abs{\K_p^{\sigma}} \leq (1 + \eta_p) F_p.
		$
		The claim follows.
	\end{proof}

	\begin{proof}[Proof of Claim~\ref{claim:low2}]
		Let $n_\tau$ be the number of distinct elements in the ground truth $\tau$, and $\{V_1, \ldots, V_{n_\tau}\}$ be the set of cliques in $G^\tau$.
		We prove for every $j \in [t]$, and thus drop the subscript $j$ in $Y_j$ for convenience. 
		
		By the definition of $Y$, we have
		$
		Y \leq m(p!) \cdot \max_{i \in [m]}  \abs{B_i^{\tau}}^{p-1} =  m(p!) \cdot \max_{k \in [n_{\tau}]}  \abs{V_k}^{p-1},
		$
		which implies 
		\begin{equation}
			\Var[Y] \leq \E[Y^2] \leq m(p!) \cdot \max_{k \in [n_{\tau}]}  \abs{V_k}^{p-1} \cdot \E[Y]. \label{eq:low2-1}
		\end{equation}
		Note that
		\begin{equation}
			\max_{k \in [n_{\tau}]}  \abs{V_k}^{p-1} = \left( \max_{k \in [n_{\tau}]} \abs{V_k}^p \right)^{1-1/p} \leq \left( \sum_{k \in [n_{\tau}]} \abs{V_k}^p \right)^{1-1/p} = F_p^{1-1/p}. \label{eq:low2-2}
		\end{equation}
		And by Hölder's inequality,
		\begin{equation}
			m = \sum_{k \in [n_{\tau}]} \abs{V_k} \leq n_{\tau}^{1-1/p} \left( \sum_{k \in [n_{\tau}]} \abs{V_k}^p \right)^{1/p} = n_{\tau}^{1-1/p} F_p^{1/p}. \label{eq:low2-3}
		\end{equation}
		Combining \eqref{eq:low2-1}, \eqref{eq:low2-2}, \eqref{eq:low2-3} and the first inequality of Claim~\ref{claim:low1}, we have
		\begin{equation*}
			\frac{\Var[Y]}{\E[Y]^2} \leq \frac{p! \cdot m \cdot \max_{k \in [n_{\tau}]}  \abs{V_k}^{p-1}}{\E[Y]} \leq \frac{p! \cdot n_{\tau}^{1-1/p} F_p^{1/p} \cdot F_p^{1-1/p}}{(1 - 2(p!) \eta_p) F_p}\leq \frac{p! \cdot m^{1-1/p} }{1 - 2(p!) \eta_p}.
		\end{equation*}
	\end{proof}

	\subsection{The Distributed Algorithm}
	\label{sec:distributed-algo}
	
	We note that our streaming algorithm cannot be directly translated into a round-efficient distributed algorithm, as sampling from the ordered $p$-cliques needs at least $p$ rounds. The key to designing a distributed algorithm that is efficient in both communication and rounds is the fact that, for the ground-truth dataset $\tau$, the quantity $F_p(\tau)$ coincides with the $(p-1)$-th {\em degree moment} of $G^{\tau}$, defined as 
	$M_{p-1}^{\tau} \triangleq \sum_{i \in [m]} \abs{B_i^{\tau}}^{p-1}.$
	When $\eta_p(\sigma, \tau)$ is small, the $(p-1)$-th degree moment $M_{p-1}^{\sigma} \triangleq \sum_{i \in [m]} \abs{B_i^{\sigma}}^{p-1}$ of $G^{\sigma}$ is expected to be close to $M^{\tau}_{p-1}$. Consequently, an accurate estimate of $M_{p-1}^{\sigma}$ yields a good approximation of $F_p(\tau)$.
	
	In the coordinator model, obtaining an unbiased estimate of $M_{p-1}^{\sigma}$ is relatively straightforward: we can uniformly sample nodes and query each site for their degrees. However, this vanilla sampling strategy suffers from instability, similar to the streaming setting. Fortunately, the analysis techniques developed for the data stream model can be applied in this context as well.
	
	Our two-round algorithm for approximating $F_p$ in the coordinator model is described in Algorithm~\ref{alg:FpDistGeneral}.  In the first round, the sites and coordinator sample a random set of nodes. At Line~\ref{ln:sample-transfer}, the coordinator tries to convert the sample from being ``without replacement'' to ``with replacement''.   In the second round, the players compute the $(p-1)$-th degree moment for each sampled node.  Finally, the coordinator use the min-average aggregation to output the final approximation.   
	
	\begin{algorithm}[t]
		\caption{Distributed-Robust-$F_p$}
		\label{alg:FpDistGeneral}
		\DontPrintSemicolon
		\SetAlgoNoEnd
		
		\KwIn{a noisy dataset $\sigma$ of size $m$ partitioned among $k$ sites, where the $\kappa$-th site holds $\sigma^{(\kappa)}$; a parameter $\eps$}
		\KwOut{an $((\eps + O(\eta_p)), 0.01)$-approximation of $F_p$}
		
		$t \gets \frac{10^6m^{1-{1}/{p}}}{\eps^2}$, $\ell \gets 100$\;  \label{ln:t-2}
		
		\gray{\tcp{First Round:}}
		Coordinator: send $m$ to all sites.
		
		\For{site $\kappa \gets 1, \dots, k$} {
			$S^{(\kappa)} \gets \emptyset$ \quad \gray{\tcc{set of sampled nodes at the $\kappa$-th site}} 
			\ForEach{$q \in \kV{\kappa}$} {
				with probability $\frac{2t \ell}{m}$, $S^{(\kappa)} \gets S^{(\kappa)} \cup \{q\}$
			}
			send $S^{(\kappa)}$ to the coordinator.
		}
		
		\gray{\tcp{Second Round:}}
		\gray{\tcc{convert ``sampled without replacement'' to ``sampled with replacement''}}
		Coordinator: $S^{\prime} \gets \bigcup_{\kappa \in [k]} S^{(\kappa)}$. If $\abs{S^{\prime}} < t \ell$, the coordinator outputs fail. Otherwise, it computes a sample $S$ from $S^{\prime}$ such that (1) $\abs{S} = t \ell$, and (2) each element in $S$ is uniformly sampled from $[m]$ with replacement. This can be done using a method described in \cite[Section 3.2]{CMY+12}. Send $S$ to all sites.   \label{ln:sample-transfer} \;
		\For{site $\kappa \gets 1, \dots, k$} {
			\ForEach{$i \in S$} {
				$d_i^{(\kappa)} \gets \abs{\kV{\kappa} \cap B_i^{\sigma} }$  \quad \gray{\tcc{local degree of node $i$ at the $\kappa$-th site}} 
			}
			send $\{d_i^{(\kappa)}\}_{i \in S}$ to the coordinator.
		}
		
		\gray{\tcp{Output:}}
		
		Coordinator: calculate $d_i \gets \sum_{\kappa \in [k]} d_i^{(\kappa)}$ for each $i \in S$. Divide $S$ into $\ell$ groups $\{S_1, \dots, S_{\ell}\}$, each containing $t$ samples. For each $j \in [\ell]$, calculate $\Bar{X}_j \gets \frac{m}{t}\sum_{i \in S_j} (d_i)^{p-1}$.
		
		\Return $\min\{\Bar{X}_1, \dots, \Bar{X}_{\ell}\}$.
	\end{algorithm}
	
	We have the following guarantees regarding Algorithm~\ref{alg:FpDistGeneral}.
	
	\begin{theorem}
		\label{thm:FpDistGeneral}
		For any constant $p \ge 1$, given an input dataset of size $m$ with $\eta_p \leq 0.4$ partitioned among $k$ sites in the coordinator model, Algorithm~\ref{alg:FpDistGeneral} computes a $((\eps + O(\eta_p)), 0.01)$-approximation of $F_p$, using two rounds and $O\left(\frac{1}{\eps^2}km^{1-1/p}\right)$ words of communication.
	\end{theorem}
	
	While Algorithm~\ref{alg:FpDistGeneral} differs considerably from Algorithm~\ref{alg:Fp-one-pass}, their analyses share similarity, especially regarding the handling of random variable instability. 
	
	\iffull
	We now prove Theorem~\ref{thm:FpDistGeneral}.
	First, by a Chernoff bound, we have $t\ell \leq \abs{S^{\prime}} \leq 3t\ell$ with probability $1-o(1)$, which implies that Algorithm~\ref{alg:FpDistGeneral} fails with probability $o(1)$.  Second, it is easy to see that Algorithm~\ref{alg:FpDistGeneral}  only uses two rounds and its communication complexity is bounded by $O(\abs{S^{\prime}} + k \abs{S}) = O(k t \ell) = O\left(\frac{1}{\eps^2}k m^{1-1/p}\right)$. It remains to prove the correctness of Algorithm~\ref{alg:FpDistGeneral}.
	
	We start with the following simple fact.
	
	\begin{fact} 
		\label{fact:deg-moments}
		For a noiseless dataset $\tau$, $F_p = \sum_{i \in [m]} (d_i^{\tau})^{p-1}$.
	\end{fact}
	
	\begin{proof}
		Let $n_\tau$ be the number of distinct elements in $\tau$, and $\{V_1, \ldots, V_{n_\tau}\}$ be the set of cliques in $G^\tau$. Then we have
		\begin{equation*}
			F_p(\tau) = \sum_{k \in [n_{\tau}]} \abs{V_k}^p = \sum_{k \in [n_{\tau}]} \sum_{i \in V_k} \abs{V_k}^{p-1} = \sum_{k \in [n_{\tau}]} \sum_{i \in V_k} \abs{B_i^{\tau}}^{p-1} = \sum_{i \in [m]} \abs{B_i^{\tau}}^{p-1}.
		\end{equation*}
	\end{proof}
	
	We have the following lemma.
	\begin{lemma} \label{lem:FpDistGeneral-exp}
		For a dataset $\sigma$, let $I \in S$ be a sampled node from Algorithm~\ref{alg:FpDistGeneral} and $X = m (d_I)^{p-1}$. We have $\E[X] = \sum_{i \in [m]} (d_i)^{p-1}$.
	\end{lemma}
	
	\begin{proof}
		Since $I$ is uniformly sampled from $[m]$, we have
		$$
		\E[X] = \sum_{i \in [m]} \Pr[I=i] \cdot m(d_i)^{p-1} = \sum_{i \in [m]} (d_i)^{p-1}.
		$$
	\end{proof}
	
	Let $X_1, \dots, X_t$ be i.i.d. copies of $X$ defined in Lemma~\ref{lem:FpDistGeneral-exp} and $\Bar{X} = \frac{1}{t} \sum_{i=1}^t X_i$.  The next two lemmas show that $\Bar{X}$ cannot be too large or too small.
	
	\begin{lemma}[$\Bar{X}$ is not too large] \label{lem:FpDistGeneral-large}
		$
		\Pr\left[\Bar{X} < (1 + \eps + \eta_p) \cdot F_p \right] > \frac{1}{13}.
		$
	\end{lemma}
	
	\begin{proof}
		By Lemma~\ref{lem:FpDistGeneral-exp}, we have $\E[\Bar{X}] = \sum_{i \in [m]} (d_i)^{p-1}$ . Applying Lemma~\ref{lem:feige} on $X_1, \ldots, X_t$, we get
		\begin{equation}
			\label{eq:apply-feige-2}
			\Pr \left[\Bar{X} < \left(1 + \frac{1}{t}\right) \sum_{i \in [m]} (d_i)^{p-1}  \right] > \frac{1}{13}.
		\end{equation}
		By Fact~\ref{fact:deg-moments}, $F_p = \sum_{i \in [m]} (d_i^{\tau})^{p-1}$. We thus have
		\begin{eqnarray}
			\sum_{i \in [m]} (d_i)^{p-1} &\leq& F_p + \abs{\sum_{i \in [m]} (d_i)^{p-1} - \sum_{i \in [m]} (d_i^{\tau})^{p-1}} \nonumber \\
			&\leq& F_p + \sum_{i \in [m]} \left(\abs{B_i^{\sigma} \cup B_i^{\tau}}^{p-1} - \abs{B_i^{\sigma} \cap B_i^{\tau}}^{p-1} \right) \nonumber \\
			&=& (1 + \eta_p) F_p. \label{eq:high1}
		\end{eqnarray}
		The lemma follows by \eqref{eq:apply-feige-2}, \eqref{eq:high1}, and the fact that $\frac{1}{t} \ll \frac{\eps}{2}$.
	\end{proof}
	
	\begin{lemma} [$\Bar{X}$ is not too small] \label{lem:FpDistGeneral-small}
		$
		\Pr\left[\Bar{X} < \left( 1 - \eps - 2\eta_p  \right) \cdot  F_p\right] < 10^{-5}.
		$
	\end{lemma}
	
	\begin{proof}
		For any $j \in [t]$, let $Y_j = \min\left\{X_j, m (d_I^{\tau})^{p-1}\right\}$. We have the following claims, whose proof will be given at the end of this section.
		
		\begin{claim} \label{claim:FpDistGeneral-low1}
			For any $j \in [t]$, $(1 - 2 \eta_p) F_p \leq \E[Y_j] \leq F_p$.
		\end{claim}
		
		\begin{claim} \label{claim:FpDistGeneral-low2}
			For any $j \in [t]$, $\frac{\Var[Y_j]}{(\E[Y_j])^2} \leq \frac{m^{1-1/p}}{1 - 2 \eta_p}$.
		\end{claim}
		
		We now have
		\begin{eqnarray}
			\Pr\left[\Bar{X} < \left( 1 - \eps - 2\eta_p \right) F_p\right] &=& \Pr\left[ \frac{1}{t} \sum_{j \in [t]} X_j - (1 - 2\eta_p)F_p < - \eps F_p \right] \label{eq:Fplow1} \\
			&\leq& \Pr\left[ \frac{1}{t} \sum_{j \in [t]} Y_j - \E[Y_j] < - \eps F_p \right] \label{eq:Fplow2} \\
			&\leq& \frac{\Var[Y_j]}{t \eps^2 F_p^2} \label{eq:Fplow21}\\
			&\leq& \frac{ \Var[Y_j]}{t \eps^2 (\E[Y_j])^2} \label{eq:Fplow3} \\
			&\leq& \frac{1 }{t \eps^2 } \cdot \frac{ m^{1-1/p}}{1 - 2 \eta_p} \label{eq:Fplow31}\\
			&\leq& 10^{-5}, \label{eq:Fplow4}
		\end{eqnarray}
		where from \eqref{eq:Fplow1} to \eqref{eq:Fplow2}, we use the fact that $X_j \geq Y_j$ from the definition of $Y_j$ and the first inequality of Claim~\ref{claim:FpDistGeneral-low1}. From \eqref{eq:Fplow2} to \eqref{eq:Fplow21}, we use Chebyshev's inequality. From \eqref{eq:Fplow21} to \eqref{eq:Fplow3}, we use the second inequality of Claim~\ref{claim:FpDistGeneral-low1}. From \eqref{eq:Fplow3} to \eqref{eq:Fplow31}, we apply Claim~\ref{claim:FpDistGeneral-low2}. The last inequality follows by our assumption that $\eta_p \leq 0.4$ and our choice of value $t$ (Line~\ref{ln:t-2} of Algorithm~\ref{alg:FpDistGeneral}).
	\end{proof}
	
	Now we are ready to prove Theorem~\ref{thm:FpDistGeneral}. Recall that we have set $\ell = 100$ in Algorithm~\ref{alg:FpDistGeneral}. Let $X_{\min} = \min\left(\Bar{X}_1, \dots, \Bar{X}_{\ell}\right)$ be the output of Algorithm~\ref{alg:FpDistGeneral}. By Lemma \ref{lem:FpDistGeneral-large} and \ref{lem:FpDistGeneral-small}, we have
	\begin{eqnarray*}
		\Pr \left[ X_{\min} < (1 + \eps + \eta_p) F_p \right] &\ge& 1 - \left( 1 - \frac{1}{13}\right)^{\ell} > 0.999, \quad \text{and} \\
		\Pr \left[ X_{\min} < (1 - \eps - 2\eta_p) F_p \right] &\le& \ell \cdot 10^{-5} = 0.001.
	\end{eqnarray*}
	The theorem follows by a union bound.
	
	Finally, we prove the two claims used in the proof of Lemma~\ref{lem:FpDistGeneral-small}. Our proofs work for every $j \in [t]$. Thus, for convenience, we drop the subscript $j$ in $X_j$ and $Y_j$. 
	
	\begin{proof}[Proof of Claim \ref{claim:FpDistGeneral-low1}]
		We first investigate the difference between $X$ and $Y$,
		\begin{eqnarray}
			\E[X - Y] &=& \sum_{i \in [m]}   \frac{1}{m} \cdot \E[X-Y \mid I = i] = \sum_{i \in [m]}  \max \left\{(d_i)^{p-1} - (d_i^{\tau})^{p-1}, 0\right\} \nonumber \\
			&\leq& \sum_{i \in [m]}  \left(\abs{B_i^{\sigma} \cup B_i^{\tau}}^{p-1} - \abs{B_i^{\sigma} \cap B_i^{\tau}}^{p-1} \right) \nonumber \\
			&=& \eta_pF_p. \label{eq:Fplow1-1}
		\end{eqnarray}
		Consequently, we have
		\begin{eqnarray*}
			\E[Y] &=& \E[X]- \E[X-Y] \\
			&\geq& \sum_{i \in [m]} (d_i)^{p-1} - \eta_pF_p \quad \quad (\text{by \eqref{eq:Fplow1-1} and Lemma~\ref{lem:FpDistGeneral-exp}}) \\
			&\geq& F_p - \abs{\sum_{i \in [m]} (d_i)^{p-1} - \sum_{i \in [m]} (d_i^{\tau})^{p-1}} - \eta_p F_p \quad \text{(by Fact~\ref{fact:deg-moments})}\\
			&\geq& (1 - \eta_p)F_p - \sum_{i \in [m]} \left(\abs{B_i^{\sigma} \cup B_i^{\tau}}^{p-1} - \abs{B_i^{\sigma} \cap B_i^{\tau}}^{p-1} \right) \\
			&=& (1-2\eta_p) F_p.
		\end{eqnarray*}
		
		For the other direction, by Fact~\ref{fact:deg-moments} we have
		\begin{eqnarray*}
			\E[Y] = \frac{1}{m} \sum_{i \in [m]} m \cdot \min \left\{(d_i)^{p-1}, (d_i^{\tau})^{p-1}\right\} \leq \sum_{i \in [m]} (d_i^{\tau})^{p-1} = F_p.
		\end{eqnarray*}
	\end{proof}
	
	\begin{proof}[Proof of Claim \ref{claim:FpDistGeneral-low2}]
		Let $n_\tau$ be the number of distinct elements in $\tau$, and $\{V_1, \ldots, V_{n_\tau}\}$ be the set of cliques in $G^\tau$.  By the definition of $Y$, we have
		\begin{equation*}
			Y \leq m \cdot \max_{i \in [m]}  \abs{B_i^{\tau}}^{p-1} =  m \cdot \max_{k \in [n_{\tau}]}  \abs{V_k}^{p-1},
		\end{equation*}
		which implies 
		\begin{equation}
			\Var[Y] \leq \E[Y^2] \leq m \cdot \max_{k \in [n_{\tau}]}  \abs{V_k}^{p-1} \cdot \E[Y]. \label{eq:Fplow2-1}
		\end{equation}
		Notice that
		\begin{equation}
			\max_{k \in [n_{\tau}]}  \abs{V_k}^{p-1} = \left( \max_{k \in [n_{\tau}]} \abs{V_k}^p \right)^{1-1/p} \leq \left( \sum_{k \in [n_{\tau}]} \abs{V_k}^p \right)^{1-1/p} = F_p^{1-1/p}. \label{eq:Fplow2-2}
		\end{equation}
		And by Hölder's inequality,
		\begin{equation}
			m = \sum_{k \in [n_{\tau}]} \abs{V_k} \leq n_{\tau}^{1-1/p} \left( \sum_{k \in [n_{\tau}]} \abs{V_k}^p \right)^{1/p} = n_{\tau}^{1-1/p} F_p^{1/p}. \label{eq:Fplow2-3}
		\end{equation}
		Combining \eqref{eq:Fplow2-1}, \eqref{eq:Fplow2-2}, \eqref{eq:Fplow2-3}, and the first inequality of Claim~\ref{claim:FpDistGeneral-low1}, we have
		\begin{equation*}
			\frac{\Var[Y]}{(\E[Y])^2} \leq \frac{m \cdot \max_{k \in [n_{\tau}]}  \abs{V_k}^{p-1}}{\E[Y]} \leq \frac{n_{\tau}^{1-1/p} F_p^{1/p} \cdot F_p^{1-1/p}}{(1 - 2 \eta_p) F_p} \leq \frac{ m^{1-1/p} }{1 - 2 \eta_p}.
		\end{equation*}
	\end{proof}
	\else
	Due to the space constraints, we defer the proof of Theorem~\ref{thm:FpDistGeneral} to the full version of this paper.
	\fi

	\section{The Lower Bounds}
	\label{sec:lb}
	
	We now explore the lower bounds, beginning with an overview of key ideas and proof techniques.
	
	\vspace{2mm}
	\noindent{\bf Ideas and Technical Overview.\ }
	We utilize the classical multiparty communication problem \kDISJ\ for proving the lower bounds.  In \kDISJ, we have $k$ players. Each player $\kappa\ (\kappa \in [k])$ is given a vector $X^{(\kappa)} = \left(X_1^{(\kappa)}, \ldots, X_n^{(\kappa)}\right) \in \{0, 1\}^n$, representing a subset of  $[n]$. We use the vector and subset representation interchangeably. The input $\{X^{(1)}, \ldots, X^{(k)}\}$ of \kDISJ\ satisfies one of the following two cases: 
	\begin{enumerate}
		\item {\em NO instance}: all $k$ subsets are pairwise disjoint. That is,  for any $\kappa \neq \kappa'$, $X^{(\kappa)} \cap X^{(\kappa')} = \emptyset$.
		
		\item {\em YES instance}: they have a unique common element but are otherwise disjoint. That is, $\abs{\bigcap_{\kappa \in [k]}  X^{(\kappa)}} = 1$ and for any $\kappa \neq \kappa'$, $X^{(\kappa)} \cap X^{(\kappa')} = \cap_{\kappa \in [k]} X^{(\kappa)}$.
	\end{enumerate}
	The problem is for the $k$ players to find out via communication whether the input is a YES instance or a NO instance. \kDISJ\ is originally studied in the blackboard model, where whenever a player sends a message, all other $(k-1)$ players can see it.  It is easy to see that any lower bound proved in the blackboard model also holds in the coordinator model, where the $k$ sites serve as the $k$ players, who communicate with each other through the coordinator.  The following result is known for \kDISJ.  We say an algorithm $\delta$-error if the algorithm errors with probability at most $\delta$.
	
	\begin{theorem}[\cite{CKS03}]
		\label{thm:kDISJ}
		Any $0.49$-error algorithm for \kDISJ\ needs $\Omega\left(\frac{n}{k\log k}\right)$ bits of communication in the blackboard model.
	\end{theorem}
	
	The \kDISJ\ problem described above was used to prove an $\tilde{\Omega}\left(m^{1-2/p}\right)$ space lower bound for the $F_p$ problem in the noiseless data stream model~\cite{CKS03}.  The standard reduction works as follows: given an  $(\eps, \delta)$-approximation streaming algorithm $\A$ for $F_p$, we can solve \kDISJ\ with parameters $n = \Theta(m)$ and $k = ((1 + 3\eps) n)^{1/p}$ in the following way: Each player $\kappa \in [k]$ constructs a set of items $\sigma^{(\kappa)}$ by including every $i \in X^{(\kappa)}$. The final stream $\sigma$ is the concatenation of all $\sigma^{(\kappa)}$, that is, $\sigma = \sigma^{(1)} \circ \ldots \circ \sigma^{(k)}$.   The $k$ players can simulate the streaming algorithm $\A$ to solve \kDISJ\ as follows: the first player runs $\A$ on $\sigma^{(1)}$ and write the end memory configuration $M_1$ to the blackboard. The second player continue runs $\A$ on $\sigma^{(2)}$ starting from $M_1$, and then write the end memory configuration $M_2$ to the blackboard, and so on. At the end, the last player outputs YES if the streaming algorithm outputs an $F_p$-estimation larger than $(1+\eps) n$, and NO otherwise. If multiple passes are allowed, then the last player can write its end memory configuration to the blackboard and restart the simulation from the first player.  If the streaming algorithm $\A$ uses $O(1)$ passes and $s$ bits of space, then the $k$ players solve \kDISJ\ problem using at most $O(s k)$ bits of communication.  By Theorem~\ref{thm:kDISJ}, we have  $s k = \Omega\left(\frac{n}{k\log k}\right)$, and consequently, $s = \Omega\left(\frac{n}{k^2\log k}\right) = \Omega\left(\frac{m^{1-2/p}}{\log m}\right)$. 
	
	However, this proof strategy is {\em not} generally applicable in the coordinator model, as the reduction requires $k = \Omega\left(m^{1/p}\right)$.  In fact, there exist $(\eps, 0.01)$-approximation algorithms for $F_p$ in the coordinator model with $k$ players using $\tilde{O}\left(\frac{k^{p-1}}{\eps^2}\right)$ bits of communication~\cite{EKM+24,HXZ+25}.  Our key observation is that in the noisy data setting, a similar reduction can be performed using only a {\em constant} number of players. Moreover, the proof only requires that the dataset has small $F_p$-mismatch-ambiguity, $\eta_p = O(\eps)$, under which an $(\eps+O(\eta_p), \delta)$-approximation is equivalent to an $(O(\eps), \delta)$-approximation.
	
	More precisely, each player $\kappa \in [k]$, for each element $i \in [n]$, adds ${t}/{k}$ pairwise {\em dissimilar} items (for an appropriate parameter $t$ to be specified later); let $Q_i^{(\kappa)}$ denote these items.  We further require that if two players $\kappa$ and $\kappa'$ share an element $i$ in the \kDISJ\ instance, then for any pair of items $u \in Q_i^{(\kappa)}$ and $v \in Q_i^{(\kappa')}$, it holds that $u \sim v$.  Conversely, if they do not share any element, then all items they create are pairwise dissimilar.  It is important to notice that this requirement cannot be satisfied in noiseless settings due to the transitivity constraint: If all items in $Q_i^{(\kappa)}$ were similar to those in $Q_i^{(\kappa')}$, then transitivity would imply that all items within $Q_i^{(\kappa)}$ to be mutually similar.  In contrast, we will show that in the noisy data setting, it is possible to construct an explicit dataset and an appropriate similarity function that jointly satisfy both requirements with probability $0.99$.  
	
	We present two concrete constructions. The first is more natural---its similarity oracle is induced by a natural distance function.  More precisely, let $s = \Theta\left(\frac{t^2\log t}{k^2}\right)$ be a parameter. For each player $\kappa \in [k]$ and each $i \in [n]$, we consider two cases (in the \kDISJ\ instance):
	\begin{enumerate}
		\item If $X_i^{(\kappa)} = 1$,  the player randomly partition the set $\{(i, k+1, 1), \ldots, (i, k+1, s)\}$ into ${t}/{k}$ subsets of equal size, and let item $q_{i,j}^{(\kappa)}$ be the $j$-th partition.  The randomness used by different players are independent. 
		
		\item If $X_i^{(\kappa)} = 0$,  the player arbitrarily partition the set $\{(i, \kappa, 1), \ldots, (i, \kappa, s)\}$ into ${t}/{k}$ subsets of equal size, and let item $q_{i,j}^{(\kappa)}$ be the $j$-th partition.  
	\end{enumerate}
	We then let $Q_i^{(\kappa)} = \left\{q_{i,1}^{(\kappa)}, \ldots, q_{i,{\frac{t}{k}}}^{(\kappa)}\right\}$. 
	
	This construction satisfies all the properties that we need. The only issue is that each item requires $O\left(m^{1/p} \log m\right)$ bits to describe, which renders a non-trivial word versus bit gap between the upper and lower bounds. We therefore propose a second, somewhat artificial construction that still satisfies all the required properties, while allowing each item to be described using only a logarithmic number of bits.  
	
	Now, let us look at the YES instance of \kDISJ. Let $i^*$ be the element shared among all players. It is easy to see that the set of items in $\sigma_Y$ created by the $k$ players from $i^*$ form a complete $k$-partitie graph, which is not far away from a clique. According to our definition of $F_p$-mismatch-ambiguity $\eta_p$, there is a noiseless dataset $\tau_Y$ such that, when used as the ground truth, the noisy dataset we construct for $F_p$ from the YES instance of \kDISJ\ has $\eta_p(\sigma_Y, \tau_Y) \le  \frac{10\eps p}{k}$, which is at most $\frac{\eps}{C+1}$ by choosing $k = 10p(C+1)$ (a constant!).
	
	On the other hand, for a NO instance of \kDISJ, the graph representation of the resulting dataset $\sigma_N$ for $F_p$ consists solely of singleton nodes, corresponding to a noiseless dataset $\tau_N$ with $\eta_p(\sigma_N, \tau_N) = 0$. By choosing the parameter $t  = (10\eps m)^{1/p}$ and $n = \frac{m}{t}$, we can show that $(1 - \eps - C\eta_p(\sigma_Y, \tau_Y))F_p(\tau_Y) > (1 + \eps + C\eta_p(\sigma_N, \tau_N))F_p(\tau_N)$.  Therefore, any $((\eps+C\eta_p), \delta)$-approximation streaming algorithm for $F_p$ that can be used to solve the \kDISJ\ problem with error at most $(\delta + 0.01)$, where the term $0.01$ counts the error probability of the input reduction.  Consequently, Theorem~\ref{thm:kDISJ} implies an $\Omega\left(\frac{1}{\eps^{1/p}} m^{1-1/p}\right)$ space lower bound for $((\eps + C\eta), 0.48)$-approximating $F_p$ in the noisy data stream.
	
	Using the same approach, we can prove a communication lower bound of $\Omega\left(\frac{1}{\eps^{1/p}} m^{1-1/p}\right)$ bits in the coordinator model.

	\vspace{2mm}
	\noindent{\bf The Lower Bounds.\ }
	We prove the following multiparty communication lower bound for $F_p$-estimation on noisy datasets. 
	
	\begin{theorem}
		\label{thm:Fp-lb}
		For any constants $p > 1$, $C \ge 0$ and $\eps \in (0, \frac{1}{3})$, any $((\eps + C\eta_p), 0.48)$-approximation algorithm for computing $F_p$ on any noisy dataset of up to $m$ items distributed over $k \ge 10p(C+1)$ players needs $\Omega\left(\frac{1}{\eps^{1/p}}m^{1-1/p}\right)$ bits of communication in the blackboard model.
	\end{theorem}
	
	This theorem yields two immediate corollaries, one for the data stream model and one for the coordinator model.  In the standard connection between the data stream model and $k$-party communication (see the technical overview), we set the number of players $k = 10p(C+1) = \Theta(1)$. 
	
	\begin{corollary}
		\label{cor:streaming-lb}
		For any constants $p > 1$, $C \ge 0$ and $\eps \in (0, \frac{1}{3})$, any $O(1)$-pass $((\eps + C\eta_p), 0.48)$-approximation algorithm for computing $F_p$ on any noisy data stream of length up to $m$ needs $\Omega\left(\frac{1}{\eps^{1/p}}m^{1-1/p}\right)$ bits of space.
	\end{corollary}
	
	\begin{corollary}
		\label{cor:distributed-lb}
		For any constants $p > 1$, $C \ge 0$ and $\eps \in (0, \frac{1}{3})$, any $((\eps + C\eta_p), 0.48)$-approximation algorithm for computing $F_p$ in the coordinator model on any noisy dataset of up to $m$ items distributed over $k \ge 10p(C+1)$ sites needs $\Omega\left(\frac{1}{\eps^{1/p}}m^{1-1/p}\right)$ bits of communication, regardless the number of rounds.
	\end{corollary}

	\iffull
	We now prove Theorem~\ref{thm:Fp-lb}. 
	As a convention, we use ``element'' to refer to an object in the input of \kDISJ, and ``item'' to refer to an object in the input of $F_p$ after the input reduction.
	
	Let $t = (10\eps m)^{1/p}$, and $n = \frac{m}{t} = (10\eps)^{-1/p} m^{1-1/p}$.  
	
	We assume $\eps \in (0, \frac{1}{3})$ for simplicity, but the proof holds for any constant $\eps \in (0, 1)$ by appropriately adjusting the constant coefficient of $t$. 
	
	\vspace{2mm}
	\noindent{\bf Input Reduction.\ }  
	We set $k = 10p(C+1)$ in \kDISJ.
	Each player $\kappa \in [k]$, for each element $i \in [n]$, creates $\frac{t}{k}$ items $Q_i^{(\kappa)} = \left\{q_{i,1}^{(\kappa)}, \ldots, q_{i,{\frac{t}{k}}}^{(\kappa)}\right\}$ for $F_p$.  Let  
	$\sigma = \cup_{i \in [n]} \cup_{\kappa \in [k]} Q_i^{(\kappa)}$
	be the resulting input for $F_p$.  Each item $q_{i,j}^{(\kappa)}$ in $\sigma$ corresponds to a node $(i, j, \kappa)$ in the graph representation $G^\sigma$.
	
	Before presenting the detailed construction of $\sigma$, we first specify its desired properties.  The explicit construction is provided at the end of this section.  
	\begin{property}
		\label{prop:input-reduction}
		The input $\sigma$ for $F_p$ after the input reduction has the following properties.
		\begin{enumerate}
			\item For any $\kappa \in [k]$ and $i \in [n]$, we have $q_{i,j}^{(\kappa)} \not\sim q_{i,j'}^{(\kappa)}$ for any $j, j' \in [\frac{t}{k}]$ ($j \neq j'$).  In words, items created from the same element w.r.t.~the same player are pairwise dissimilar. 
			
			\item For any $i \in [n]$ and any $\kappa, \kappa' \in [k]\ (\kappa \neq \kappa')$  with $X_i^{(\kappa)} = X_i^{(\kappa')} = 1$, we have $q_{i,j}^{(\kappa)} \sim q_{i,j'}^{(\kappa')}$ for any $j, j' \in [\frac{t}{k}]$.   In words, if two players $\kappa, \kappa'$ share an element $i$ (i.e., $X_i^{(\kappa)} = X_i^{(\kappa')} = 1$), then any item created by player $\kappa$ from $i$ is similar to any item created by player $\kappa'$ from $i$.
			
			\item For any $i \in [n]$ and any $\kappa, \kappa' \in [k]\ (\kappa \neq \kappa')$  with $X_i^{(\kappa)} = 0$ or $ X_i^{(\kappa')} = 0$, we have $q_{i,j}^{(\kappa)} \not\sim q_{i,j'}^{(\kappa')}$ for any $j, j' \in [\frac{t}{k}]$.   In words, an item created from a zero entry (i.e., $X_i^{(\kappa)} = 0$) is dissimilar to any other items created by the same element.
			
			\item For any $i, i' \in [n]\ (i \neq i')$, we have  $q_{i,j}^{(\kappa)} \sim q_{i,j'}^{(\kappa')}$ for any $j, j' \in [\frac{t}{k}]$ and $\kappa, \kappa' \in [k]$.  In words, items created from different elements are pairwise dissimilar.
		\end{enumerate}
	\end{property}
	
	It is easy to see that each \kDISJ\ element $i \in [n]$ held by only one player is mapped into $\frac{t}{k} \cdot k = t$ singleton nodes in $G^\sigma$, while the element held by all players (if exists) is mapped into a complete $k$-partite graph in $G^\sigma$ with each of the $k$ players holding an independent set of size $\frac{t}{k}$.
	
	\vspace{2mm}
	\noindent{\bf Mismatch Ambiguity of $\sigma$.\ }
	We now investigate the properties of $\sigma$ when it is generated by either a YES instance or a NO instance of \kDISJ.
	
	Let $\sigma_N$ be the resulting dataset for $F_p$ after the input reduction from a NO \kDISJ\ instance. Considering the ground truth $\tau_N$ that coincides $\sigma_N$, we have
	\begin{equation}
		\label{eq:tau_N}
		\eta_p(\sigma_N, \tau_N) = 0, \quad \text{and} \quad F_p(\tau_N) = n \cdot t = m.
	\end{equation}
	
	Let $\sigma_Y$ be the resulting dataset for $F_p$ after the input reduction from a YES \kDISJ\ instance.  We consider the following ground truth dataset $\tau_Y$ with $G^{\tau_Y}$: 
	For the element $i^* \in [n]$ that is held by all $k$ players, we create clique consisting of nodes $\left\{(i^*, j, \kappa)\ |\ j \in [\frac{t}{k}], \kappa \in [k]\right\}$.  For each $i \in [n]\backslash\{i^*\}$, we create a set of singleton nodes $\left\{(i, j, \kappa)\ |\ j \in [\frac{t}{k}], \kappa \in [k]\right\}$.
	
	We have 
	\begin{equation}
		\label{eq:tau_Y}
		F_p(\tau_Y) = (n-1) \cdot t + t^p \in ((1 + 9\eps)m, (1 + 10\eps)m) .
	\end{equation}

	To compute $\eta_p(\sigma_Y, \tau_Y)$, we notice that the only difference between $\sigma_Y$ and $\tau_Y$ comes from items indexed by $\{(i^*, j, \kappa)\ |\ j \in [\frac{t}{k}], \kappa \in [k]\}$, or, the complete $k$-partite graph in $G^{\sigma_Y}$ versus the clique in $G^{\tau_Y}$.  We have
	\begin{eqnarray}
		\eta_p(\sigma_Y, \tau_Y) &=& \frac{1}{F_p(\tau_Y)} \sum_{\kappa \in [k]}\sum_{j \in [\frac{t}{k}]}  \left(\abs{B_{(i^*, j, \kappa)}^{\sigma} \cup B_{(i^*, j, \kappa)}^{\tau}}^{p-1} - \abs{B_{(i^*, j, \kappa)}^{\sigma} \cap B_{(i^*, j, \kappa)}^{\tau}}^{p-1} \right) \nonumber \\
		&\stackrel{\eqref{eq:tau_Y}}{\le}& \frac{1}{m} \cdot k \cdot \frac{t}{k} \cdot \left(t^{p-1} - \left(t - \frac{t}{k} + 1\right)^{p-1}\right)\nonumber \\
		&\le& \frac{1}{m} \cdot t^p \cdot \left(1 - \left(1 - \frac{1}{k}\right)^{p-1}\right) 
		\le  \frac{1}{m} \cdot t^p \cdot \left(1 - \left(1 - \frac{1}{k}\right)^p\right) \nonumber \\
		&\le& 10\eps \cdot \frac{p}{k} \label{eq:ambiguity-threshold} \\
		&\le& \frac{\eps}{C+1}, \label{eq:eta-yes} 
	\end{eqnarray} 
	where the last inequality holds since we have set $k = 10p(C+1)$.
	
	\vspace{2mm}
	\noindent{\bf Reducing \kDISJ\ to $F_p$.\ } 
	Suppose we have a $(\eps + C\eta_p, \delta)$-approximation algorithm $\A$ for $F_p$, we construct an algorithm $\A'$ for \kDISJ\ as follows:  giving an input $I$ for \kDISJ, we use our input reduction to create an input $\sigma$ for $F_p$.  We then run $\A$ on $\sigma$ and get an estimate $\tilde{F}_p$.  If $I$ is a YES instance, then by \eqref{eq:tau_Y} and \eqref{eq:eta-yes}, we have  with probability $(1-\delta-o(1))$ (the extra $o(1)$ error comes from the randomized input reduction, as it will be evident from our explicit constructions shortly),
	\begin{equation*}
		\tilde{F}_p \ge (1 - \eps - C\eta_p(\sigma_Y, \tau_Y)) F_p(\tau_Y)  \ge (1-2\eps) (1 + 9\eps)m > (1+\eps)m.
	\end{equation*}
	If $I$ is a NO instance, then by \eqref{eq:tau_N}, we have with probability $(1-\delta-o(1))$,
	\begin{equation*}
		\tilde{F}_p \le (1 + \eps + C \eta_p(\sigma_N, \tau_N)) F_p(\tau_N)  \le (1+\eps)m.
	\end{equation*}
	
	Therefore, by checking whether $\tilde{F}_p > (1+\eps) m$, we can decide whether \kDISJ$(I) = 1$ or $0$. Consequently, any $((\eps + C\eta_p), \delta)$-approximation algorithm for $F_p$ can be used to solve \kDISJ\ with probability $(1-\delta-o(1))$. By Theorem~\ref{thm:kDISJ}, we have that any algorithm that $((\eps + C\eta_p), 0.48)$-approximation algorithm for $F_p$ needs at least $\frac{n}{k \log k} = \Omega\left(\frac{1}{\eps^{1/p}} m^{1-1/p}\right)$ bits of communication.

	\subsection{Explicit Constructions of $\sigma$}
	\label{sec:input-construct}
	We now give some explicit constructions of $\sigma = \cup_{i \in [n]} \cup_{\kappa \in [k]} Q_i^{(\kappa)}$, where $Q_i^{(\kappa)} = \left\{q_{i,1}^{(\kappa)}, \ldots, q_{i,{\frac{t}{k}}}^{(\kappa)}\right\}$, that satisfies Property~\ref{prop:input-reduction}.
	
	We begin with a construction where similarity is defined using a natural distance function---the set symmetric difference. A drawback of this approach is that each item (a set) requires  $b = O\left(m^{1/p}\log m\right)$
	bits to represent, creating a gap of $b$ between the upper and lower bounds in terms of $m$ if both are measured by bits.\footnote{In the noiseless setting, item size does not affect the performance of subsampling-based algorithms, provided that we can subsample each item using a hash function. However, in the noisy setting, we must store all items in full, since the similarity function may require access to their complete descriptions.} To address this, we present a second construction in which each item can be represented using only $O(\log m)$ bits, making the upper and lower bounds nearly tight in terms of $m$. The tradeoff is that the similarity function in this case is not derived from a natural distance function. Nevertheless, our algorithms apply to any dataset~$\sigma$, as long as a binary similarity function can be defined on~$\sigma$.
	
	\vspace{2mm}
	\noindent{\bf A Construction Based on Set Symmetric Differences.\ }
	Let $s = \frac{3t^2\log t}{k^2}$. For each player $\kappa \in [k]$ and each $i \in [n]$, we consider two cases:
	\begin{enumerate}
		\item If $X_i^{(\kappa)} = 1$,  the player randomly partition the set $\{(i, k+1, 1), \ldots, (i, k+1, s)\}$ into $\frac{t}{k}$ subsets of equal size, and let item $q_{i,j}^{(\kappa)}$ be the $j$-th partition.  The randomness used by different players are independent. 
		
		\item If $X_i^{(\kappa)} = 0$,  the player arbitrarily partition the set $\{(i, \kappa, 1), \ldots, (i, \kappa, s)\}$ into $\frac{t}{k}$ subsets of equal size, and let item $q_{i,j}^{(\kappa)}$ be the $j$-th partition.  
	\end{enumerate}
	Each item can be described using $O\left(s/\frac{t}{k} \cdot \log (nks)\right) = O\left(m^{1/p} \log m\right)$ bits.  
	
	We now check whether $\sigma$ satisfies the four requirements in Property~\ref{prop:input-reduction}.
	
	First, due to the set partition in our construction, it is clear that for any $\kappa \in [k]$, $i \in [n]$, and $j, j' \in [\frac{t}{k}]$, we have $\abs{q_{i,j}^{(\kappa)} \cap q_{i,j'}^{(\kappa)}} = 0$.
	
	Second, suppose two players $\kappa, \kappa' \in [k]\ (\kappa \neq \kappa')$ share the same element $i^*$.  Set $z = \frac{sk}{t}$. For $j, j' \in [\frac{t}{k}]$, we have 
	\begin{eqnarray*}
		\Pr\left[\abs{q_{i,j}^{(\kappa)} \cap q_{i,j'}^{(\kappa)}} \ge 1\right] & = & 1 - \frac{\binom{s-z}{z}}{\binom{s}{z}} =  1 - \frac{(s-z)! / (z!(s-2z)!)}{s! / (z!(s-z)!)} \\
		&=& 1 - \frac{(s-z)(s-z-1)\cdots(s-2z+1)}{s(s-1)\cdots(s-z+1)} \\
		&=& 1 - \prod_{a=0}^{z-1} \left(1 - \frac{z}{s-a}\right) \\
		&=& 1 - \exp\left(\sum_{a=0}^{z-1} \ln\left(1 - \frac{z}{s-a}\right) \right) \\
		&\ge& 1 - \exp\left(\sum_{a=0}^{z-1} -\frac{z}{s-a} \right) \\
		&\ge&  1 - \exp\left(-\frac{z^2}{s}\right) = 1 - \frac{1}{t^3}.
	\end{eqnarray*}
	By a union bound, with probability at least $1 - t^2\cdot  \frac{1}{t^3} = 1 - o(1)$, for all $\kappa, \kappa' \in [k]\ (\kappa \neq \kappa')$ and $j, j' \in [\frac{t}{k}]$, we have $\abs{q_{i^*,j}^{(\kappa)} \cap q_{i^*,j'}^{(\kappa')}} \ge 1$.  
	
	Third, for any $i \in [n]$ and any $\kappa, \kappa' \in [k]$ with $\kappa \neq \kappa'$, 
	if $X_i^{(\kappa)} = 0$ or $X_i^{(\kappa')} = 0$, then our construction guarantees that 
	the second entries of all tuples in $q_{i,j}^{(\kappa)}$ are distinct from those in $q_{i,j'}^{(\kappa')}$ for any $j, j' \in [\frac{t}{k}]$. 
	Therefore, $\abs{q_{i,j}^{(\kappa)} \cap q_{i,j'}^{(\kappa')}} = 0$.
	
	Finally, for any $i, i' \in [n]\ (i \neq i')$, it is clear that in our construction, the first entries of all tuples in $q_{i,j}^{(\kappa)}$ are distinct from those in $q_{i,j'}^{(\kappa')}$ for any $j, j' \in [\frac{t}{k}]$ and $\kappa, \kappa' \in [k]$. Therefore, $\abs{q_{i,j}^{(\kappa)} \cap q_{i',j'}^{(\kappa')}} = 0$.
	
	Now, for two sets $u, v \in \sigma$, if we define $u \sim v$ when $\abs{u \cap v} \ge 1$ (or, the symmetric difference $\abs{u \oplus v} \le \frac{2sk}{t} - 2)$, and  $u \not\sim v$ otherwise, then the resulting input for robust $F_p$ is able to satisfy all requirements in Property~\ref{prop:input-reduction} with probability $(1 - o(1))$. 
	
	\vspace{2mm}
	\noindent{\bf A Construction with Compact Item Representations.\ }  We now give a construction in which each item can be described in $O(\log m)$ bits.  This construction is in fact simpler than the previous one, but its similarity function is not derived from a distance metric. 
	
	For each player $\kappa \in [k]$ and each $i \in [n]$, we again consider two cases: 
	\begin{enumerate}
		\item If $X_i^{(\kappa)} = 1$,  for each $j \in [\frac{t}{k}]$, $q_{i,j}^{(\kappa)} = (i, k+1, jn + Y_\kappa)$, where $Y_\kappa$ is sampled from $\{1, \ldots, n\}$ uniformly at random.   
		
		\item If $X_i^{(\kappa)} = 0$,  for each $j \in [\frac{t}{k}]$, $q_{i,j}^{(\kappa)} = (i, \kappa, jn + Y_\kappa)$, where $Y_\kappa$ is sampled from $\{1, \ldots, n\}$ uniformly at random.
	\end{enumerate}
	We then define the similarity between two items $(a, b, c)$ and $(a', b', c')$ as follows: 
	\[
	(a,b,c) \sim (a',b',c') \;=\;
	\begin{cases}
		\text{true}, & \text{if } (a = a') \wedge (b = b') \wedge (c = c'), \\[6pt]
		\text{true}, & \text{if } (a = a') \wedge (b = b') \wedge (c \neq c' \pmod n), \\[6pt]
		\text{false}, & \text{otherwise.}
	\end{cases}
	\]

	We now check the four requirements in Property~\ref{prop:input-reduction}. 
	The third and fourth requirements are straightforward, as in each case the conditions ensure that the two items under consideration have distinct first or second entries. The first requirement is satisfied since $Y_\kappa$ is the same for the constructions of all nodes in the $\kappa$-th site. For the second requirement, consider the set of items $Q_{i^*} = \{q_{i^*,j}^{(\kappa)}\ |\ \kappa \in [k], j \in [\frac{t}{k}]\} = \{(i^*, k+1, jm + Y_\kappa)\ |\ \kappa \in [k], j \in [\frac{t}{k}]\}$, where $i^*$ is the unique common element in a YES instance of \kDISJ. Since all $Y_\kappa$ are sampled randomly from $\{1, \ldots, n\}$, the probability that there exists a pair of identical items in $\{Y_1, \ldots, Y_k\}$ is upper bounded by
	$
	\frac{k^2}{n} = o(1).
	$
	Therefore, with probability $(1 - o(1))$, the alternative construction also satisfies all requirements in Property~\ref{prop:input-reduction}.
	
	\else
	
	Due to the space constraints, we defer the proof of Theorem~\ref{thm:Fp-lb} to the full version of this paper.
	
	\fi

	\section{An Improved Distributed Algorithm for Low-Ambiguity Regime}
	\label{sec:small-eta-p}
	
	As noted in the introduction, the communication cost of Algorithm~\ref{alg:FpDistGeneral} presented in Section~\ref{alg:FpDistGeneral} is asymptotically optimal in terms of $m$. However, we find that when the mismatch ambiguity $\eta_p$ falls below a certain threshold, it is possible to design algorithms with communication cost {\em independent} of $m$.   In this section, we present a distributed algorithm in the coordinator model for datasets with small $F_p$-mismatch ambiguity, achieving a communication cost of $\text{poly}\left(k, \frac{1}{\eps}\right)$ words. 
	
	\vspace{2mm}
	\noindent{\bf Ideas and Technical Overview.\ }
	For ease of reference, we summarize in Table~\ref{tab:notation} a set of notations that will be used throughout the remainder of this section.
	
	\begin{table}[t]
		\centering
		\renewcommand{\arraystretch}{1.3}
		\begin{tabular}{>{$}l<{$} p{0.5\linewidth}}
			\hline
			\textbf{Notation} & \textbf{Description} \\
			\hline
			\kV{\kappa} & subset of $[m]$ held by site $\kappa \in [k]$. \\
			
			d_i^{(\kappa)} \triangleq \abs{\kV{\kappa} \cap B_i^{\sigma}} 
			& number of nodes similar to $\sigma_i$ held by site $\kappa$. \\
			
			d_i \triangleq \sum_{\kappa \in [k]} d_i^{(\kappa)} = \abs{B_i^\sigma} 
			& global degree of $\sigma_i$. \\
			
			\dl{i} \triangleq d_{i}^{(\kappa(i))} 
			& local degree of $\sigma_i$, where $\kappa(i) \in [k]$ is the site that holds $\sigma_i$. \\
			
			\Dl{\kappa} \triangleq \sum_{i \in \kV{\kappa}} (\dl{i})^{p-1} 
			& sum of local degree moments of $\kV{\kappa}$. \\
			
			D_{\loc} \triangleq \sum_{\kappa \in [k]} \Dl{\kappa} = \sum_{i \in [m]} (\dl{i})^{p-1} 
			& sum of local degree moments of $[m]$. \\
			
			\dl{i}[\tau] \triangleq \abs{B_i^{\tau} \cap \kV{\kappa(i)}} 
			& local degree of $\tau_i$. \\
			
			d_i^{\tau} \triangleq \abs{B_i^{\tau}} 
			& global degree of $\tau_i$. \\
			\hline
		\end{tabular}
		\caption{Summary of notations.}
		\label{tab:notation}
	\end{table}

	Same as that in the general case (Section~\ref{sec:distributed-algo}), our goal is to estimate the $(p-1)$-th degree moment $M_{p-1}^{\sigma}$, which is close to $F_p$ when $\eta_p$ is small. 
	Intuitively, we can view $\eta_p$ as a budget available to an adversary, who perturbs the ground truth graph $G^\tau$ by adding or removing edges.  Consider, for example, a graph consisting of $(m - m^{1/p})$ isolated nodes and a clique of size $m^{1/p}$, whose nodes are evenly distributed across $k$ sites. When $\eta_p = \Theta(\eps)$, the adversary can use their budget to convert the clique into a $k$-partite graph by deleting all edges with both end nodes located at the same site. As a result, each site cannot locally determine whether a node is part of a clique or merely an isolated node---yet this distinction is critical due to the substantial contribution of clique nodes to the quantity $M_{p-1}^{\sigma}$.
	
	Our main idea for datasets with small $\eta_p$ is that when the adversary does not have sufficient budget to eliminate all local structural information, then we can exploit the remaining local structures to identify high-degree nodes with a small amount of communication.  Note that the contribution of high-degree nodes dominate the value of $M_{p-1}^{\sigma}$.
	
	We first consider the noiseless case to see how the local information helps.  Let $\sigma$ be a noiseless dataset, and let $V$ denote the largest clique of $G^{\sigma}$.  For each $\kappa \in [k]$, let $\kV{\kappa}$ be the nodes in $V$ that is held by site $\kappa$. Let $\alpha = \frac{k}{\eps}$ be a threshold.  If $\frac{\abs{V}}{\abs{\kV{\kappa}}} > \alpha$, then $\kV{\kappa}$ is small and the site $\kappa$ may not be able to identify $\kV{\kappa}$ as part of a large clique $V$ without communication.  But this is not a big issue, since the contribution of $\kV{\kappa}$ is at most $\frac{1}{\alpha}$-fraction of that of $V$ and can thus be ignored.  On the other hand, if $\frac{\abs{V}}{\abs{\kV{\kappa}}} \leq \alpha$, then $\abs{\kV{\kappa}}$ is large given $\abs{V}$ is large. In this case, site $\kappa$ can identify $V^{(\kappa)}$ as part of a large clique without any communication. 
	Therefore, to estimate $M_{p-1}^\sigma$, it suffices to estimate the contribution of large $V^{(\kappa)}$s.

	To implement this idea, we sample nodes according to the probabilities $\left\{\frac{(\dl{i})^{p-1}}{D_{\loc}} \right\}_{i \in [m]}$ (i.e., proportional to the $(p-1)$-th power of their local degrees; definitions for $\dl{i}$, $D_{\loc}$, and other newly introduced notations in this overview are provided in Table~\ref{tab:notation}). For each sampled node $I$, we query all $k$ sites to obtain its global degree $d_I$. It is easy to verify that $Y = D_{\loc} \cdot\left(\frac{d_I}{\dl{I}}\right)^{p-1}$ serves as an unbiased estimator for $M_{p-1}^{\sigma}$. However, this estimator is unstable due to its high variance. To resolve this issue, we introduce a truncated estimator $X = Y \cdot \mathbf{1}\left\{ \frac{d_I}{\dl{I}} \leq \alpha \right\}$. $X$ is an biased estimator of $M^\sigma_{p-1}$, where the bias comes from the nodes whose global-to-local degree ratio exceeds $\alpha$ (i.e., the nodes from small $\kV{\kappa}$s), which counts for at most a $\frac{k}{\alpha} (= \eps)$ fraction of $M^\sigma_{p-1}$.  We can further show that $O\left(\frac{\alpha^{p-1}}{\eps^2}\right) = O\left(\frac{k^{p-1}}{\eps^{p+1}}\right)$ samples are enough for an accurate approximation of $\E[X]$.

	Now consider the case that the dataset has a small ambiguity. For our method to fail, the adversary must select some nodes from large $\kV{\kappa}$s and change their global-to-local degree ratios by adding or removing edges, in order to push them above the threshold $\alpha$.  To this end, we show that the adversary’s most effective strategy is to remove local edges.  As a simple example, consider the case when $p = 2$ and a node $i$ such that ${d_i^{\tau}}/{\dl{i}[\tau]} = {\alpha}/{2}$.  To increase ${d_i^{\tau}}/{\dl{i}[\tau]}$ to $\alpha$, the adversary needs to remove roughly ${\dl{i}[\tau]}/{2} = {d_i^{\tau}}/{\alpha}$ local edges incident to node $i$, which eliminates node $i$ from the set of low global-to-local ratio nodes. Since the contribution of node $i$ to $F_2$ is $d_i \le d_i^\tau$, this operation will only increase the bias of $X$ by at most $d_i^\tau$.  Therefore, given a budget of $\eta_2 F_2$ edge insertions/deletions, the adversary can only increase the bias of the input by at most $\alpha \eta_2 F_2$, which is at most $\eps F_2$ when $\eta_2 \le \frac{\eps}{\alpha} = O\left(\frac{\eps^2}{k}\right)$.  We can extend this argument to the general setting and any $p \ge 2$, showing that the bias of $X$ is at most $\eps F_p$ when $\eta_p \le O\left(\frac{\eps^p}{k^{p-1}}\right)$.

	\vspace{2mm}
	\noindent{\bf Algorithm and Analysis.\ }
	The algorithm is described in Algorithm~\ref{alg:FpDist}.  
	Let us briefly describe it in words.  In the first round, each site sends the sum of the local $(p-1)$-th degree moments of its nodes to the coordinator, who aggregates them to obtain $D_{\loc}$, the total $(p-1)$-th degree moments over all $m$ nodes. In the second round, the sites and the coordinator jointly sample each node with probability proportional to its local $(p-1)$-th degree moment. In the third round, the exact global degree of each sampled node is computed. Finally, using $D_{\loc}$ together with the local and global $(p-1)$-th degree moments of the sampled nodes, the coordinator constructs the final truncated estimator.

	\begin{algorithm}[!th]
		\caption{Distributed-Robust-$F_p$-Low-Noise}
		\label{alg:FpDist}
		\DontPrintSemicolon
		\SetAlgoNoEnd
		
		\KwIn{a noisy dataset $\sigma$ of size $m$ partitioned among $k$ sites, where the $\kappa$-th site holds $\sigma^{(\kappa)}$; a parameter $\eps$}
		\KwOut{an $((\eps + O(\eta_p)), 0.01)$-approximation of $F_p$}
		
		$t \gets \frac{36\cdot 4^p \cdot k^{p-1}}{\eps^{p+1}}$, $\theta \gets 2 \left(\frac{4k}{\eps}\right)^{p-1}$
		
		\gray{\tcp{First Round:}}
		Coordinator: send ``start'' to all sites.
		
		\For{site $\kappa \gets 1, \dots, k$} {
			$\Dl{\kappa} \gets \sum_{i \in \kV{\kappa}} (\dl{i})^{p-1}$\;
			send $\Dl{\kappa}$ to the coordinator
		}
		
		\gray{\tcp{Second Round:}}
		Coordinator: calculate $D_{\loc} \gets \sum_{\kappa \in [k]} \Dl{\kappa}$. Then, pick $t$ i.i.d.\ samples from $[k]$ according to probabilities $\left\{\frac{\Dl{\kappa}}{D_{\loc}}\right\}_{\kappa \in [k]}$. Let $s_{\kappa}$ be the number of times $\kappa$ is sampled. Send $s_{\kappa}$ to site $\kappa$ for each $\kappa \in [k]$.
		
		\For{site $\kappa \gets 1, \dots, k$} {
			pick $s_{\kappa}$ i.i.d.\ samples from $\kV{\kappa}$ according to probabilities $\left\{\frac{(\dl{i})^{p-1}}{\Dl{\kappa}}\right\}_{i \in V^{(\kappa)}}$\;
			send all the sampled nodes to the coordinator
		}

		\gray{\tcp{Third Round:}}
		Coordinator: let $S$ be the set of sampled nodes. Send $S$ to all sites.\;
		\For{site $\kappa \gets 1, \dots, k$} {
			\ForEach{$i \in S$} {
				$d_i^{(\kappa)} \gets \abs{\kV{\kappa} \cap B_i^{\sigma} }$
			}
			send $\{d_i^{(\kappa)}\}_{i \in S}$ to the coordinator
		}
		
		\gray{\tcp{Output:}}
		
		Coordinator: calculate $d_i \gets \sum_{\kappa \in [k]} d_i^{(\kappa)}$ for each $i \in S$.
		
		\Return $\frac{1}{t} \sum_{i \in S} \left(D_{\loc} \cdot \left(\frac{d_i}{\dl{i}}\right)^{p-1} \cdot \mathbf{1}\left\{ \left(\frac{d_i}{\dl{i}}\right)^{p-1} \leq \theta \right\}\right)$.  \label{ln:output}
	\end{algorithm}
	
	We have the following result regarding Algorithm~\ref{alg:FpDist}.  
	
	\begin{theorem}
		\label{thm:FpDist}
		For any constant $p \ge 2$ and a noisy size-$m$ dataset $\sigma$ with $\eta_p \leq \frac{\eps^p}{4^{p+1} \cdot k^{p-1}}$, Algorithm~\ref{alg:FpDist} computes an $((\eps + O(\eta_p)), 0.01)$-approximation of $F_p$ in the coordinator model with $k$ sites, using three rounds and $O\left(\frac{k^p}{\eps^{p+1}}\right)$ words of communication.
	\end{theorem} 
	
	In the rest of this section, we prove Theorem~\ref{thm:FpDist}. The round complexity is clear from Algorithm~\ref{alg:FpDist}. The communication cost is dominated by communicating $S$ between sites and the coordinator, which uses $O(kt) = O(k^p / \eps^{p+1})$ words.
	In the rest of this section, we prove the correctness of the algorithm.  
	
	\iffull
	\else
	We start with the following simple fact.
	
	\begin{fact} 
		\label{fact:deg-moments}
		For a noiseless dataset $\tau$, $F_p = \sum_{i \in [m]} (d_i^{\tau})^{p-1}$.
	\end{fact}
	
	\begin{proof}
		Let $n_\tau$ be the number of distinct elements in $\tau$, and $\{V_1, \ldots, V_{n_\tau}\}$ be the set of cliques in $G^\tau$. Then we have
		\begin{equation*}
			F_p(\tau) = \sum_{k \in [n_{\tau}]} \abs{V_k}^p = \sum_{k \in [n_{\tau}]} \sum_{i \in V_k} \abs{V_k}^{p-1} = \sum_{k \in [n_{\tau}]} \sum_{i \in V_k} \abs{B_i^{\tau}}^{p-1} = \sum_{i \in [m]} \abs{B_i^{\tau}}^{p-1}.
		\end{equation*}
	\end{proof}
	\fi
	
	The following lemma shows that each summand at Line~\ref{ln:output} of Algorithm~\ref{alg:FpDist} does not deviate from $F_p$ by much.
	\begin{lemma} \label{lem:FpDist-exp}
		For a dataset $\sigma$, let $I \in S$ be a sampled node from Algorithm~\ref{alg:FpDist} and $X = D_{\loc} \cdot \left(\frac{d_I}{\dl{I}}\right)^{p-1} \cdot \mathbf{1}\left\{ \left(\frac{d_I}{\dl{I}}\right)^{p-1} \leq \theta \right\}$. We have
		$
		\left(1 - 2\eta_p - \frac{\eps}{2}\right)F_p \leq \E[X] \leq (1 + \eta_p) F_p.
		$
	\end{lemma}
	
	\begin{proof}
		We start with the second inequality.  Let $p_i = \Pr[I = i]$. We have
		\begin{eqnarray*}
			p_i &=& \Pr[\text{site } \kappa(i) \text{ is sampled}]\cdot \Pr[i \text{ is sampled} \mid \text{site } \kappa(i) \text{ is sampled}] \\
			&=& \frac{\Dl{\kappa(i)}}{D_{\loc}} \cdot \frac{(\dl{i})^{p-1}}{\Dl{\kappa(i)}} = \frac{(\dl{i})^{p-1}}{D_{\loc}}.
		\end{eqnarray*}
		Let $Y = D_{\loc} \cdot \left(\frac{d_I}{\dl{I}}\right)^{p-1}$. By the definition of $X$, we have
		\begin{equation}
			\E[X] \leq \E[Y] = \sum_{i \in [m]} p_i \cdot D_{\loc} \cdot \left(\frac{d_i}{\dl{i}}\right)^{p-1} = \sum_{i \in [m]} (d_i)^{p-1}. \label{eq:exp1}
		\end{equation}
		Note that
		\begin{equation}
			\abs{\sum_{i \in [m]} (d_i)^{p-1} - \sum_{i \in [m]} (d_i^{\tau})^{p-1}} \leq \sum_{i \in [m]} \left(\abs{B_i^{\sigma} \cup B_i^{\tau}}^{p-1} - \abs{B_i^{\sigma} \cap B_i^{\tau}}^{p-1} \right) = \eta_pF_p. \label{eq:exp2}
		\end{equation}
		By \eqref{eq:exp2} and Fact~\ref{fact:deg-moments}, we have
		\begin{equation}
			(1 - \eta_p) F_p \leq \sum_{i \in [m]} (d_i)^{p-1} \leq (1+\eta_p) F_p. \label{eq:exp3}
		\end{equation}
		Combining \eqref{eq:exp1} and \eqref{eq:exp3}, it follows that $\E[X] \leq (1 + \eta_p)F_p$.
		
		We now prove the first inequality.   We start by bounding the difference between $X$ and $Y$.
		\begin{eqnarray*}
			\E[Y-X] &=& \sum_{i \in [m]} p_i \cdot D_{\loc} \cdot \left(\frac{d_i}{\dl{i}}\right)^{p-1} \cdot \mathbf{1}\left\{ \left(\frac{d_i}{\dl{i}}\right)^{p-1} > \theta \right\} \\
			&=& \sum_{i \in [m]} (d_i)^{p-1} \cdot \mathbf{1}\left\{ \left(\frac{d_i}{\dl{i}}\right)^{p-1} > \theta \right\} \\
			&\leq& \sum_{i \in [m]} \abs{B_i^{\sigma} \cup B_i^{\tau}}^{p-1} \cdot \mathbf{1}\left\{ \left(\frac{d_i}{\dl{i}}\right)^{p-1} > \theta \right\} \\
			&\leq& \underbrace{\sum_{i \in [m]} \left(\abs{B_i^{\sigma} \cup B_i^{\tau}}^{p-1} - (d_i^{\tau})^{p-1}\right)}_{\textcircled{1}} + \underbrace{\sum_{i \in [m]} (d_i^{\tau})^{p-1} \cdot \mathbf{1}\left\{ \left(\frac{d_i}{\dl{i}}\right)^{p-1} > \theta \right\}}_{\textcircled{2}}.
		\end{eqnarray*}
		For the first part,
		\begin{eqnarray}
			\textcircled{1} \leq \sum_{i \in [m]} \left(\abs{B_i^{\sigma} \cup B_i^{\tau}}^{p-1} - \abs{B_i^{\sigma} \cap B_i^{\tau}}^{p-1}\right) = \eta_p F_p. \label{eq:exp4}
		\end{eqnarray}
		For the second part, we divide $[m]$ into two sets using a threshold $\alpha = \left(\frac{\theta}{2}\right)^{\frac{1}{p-1}}$: 
		$$
		H = \left\{i \in [m]\ \left|\ \frac{d_i^{\tau}}{\dl{i}[\tau]} > \alpha \right.\right\} \quad \text{and} \quad L = \left\{i \in [m]\ \left|\ \frac{d_i^{\tau}}{\dl{i}[\tau]} \leq \alpha \right. \right\}.
		$$
		We utilize the following two claims to bound the contributions of items in $H$ and $L$, respectively. Their proofs will be given shortly.
		\begin{claim} \label{claim:small}
			\begin{equation*}
				\sum_{i \in H} (d_i^{\tau})^{p-1} \leq \frac{k F_p}{\alpha}.
			\end{equation*}
		\end{claim}
		\begin{claim} \label{claim:large}
			\begin{equation*}
				\sum_{i \in L} (d_i^{\tau})^{p-1} \cdot \mathbf{1}\left\{ \left(\frac{d_i}{\dl{i}}\right)^{p-1} > \theta \right\} \leq \theta \eta_p F_p.
			\end{equation*}
		\end{claim}
		By Claim~\ref{claim:small} and Claim~\ref{claim:large}, we have
		\begin{eqnarray}
			\textcircled{2} &\leq& \sum_{i \in H} (d_i^{\tau})^{p-1} + \sum_{i \in L} (d_i^{\tau})^{p-1} \cdot \mathbf{1}\left\{ \left(\frac{d_i}{\dl{i}}\right)^{p-1} > \theta \right\} \nonumber \\
			&\leq& \frac{k F_p}{\alpha} + \theta \eta_p F_p \leq \frac{\eps F_p}{2}, \label{eq:exp5}
		\end{eqnarray}
		where in the last inequality, we have used $\theta = 2\left(\frac{4k}{\eps}\right)^{p-1}$, $\alpha = \left(\frac{\theta}{2}\right)^{\frac{1}{p-1}} = \frac{4k}{\eps}$, and our assumption $\eta_p \leq \frac{\eps^p}{4^{p+1} \cdot k^{p-1}}$.
		Combining \eqref{eq:exp4} and \eqref{eq:exp5}, we have
		\begin{eqnarray}
			\E[Y-X] \leq \eta_p F_p + \frac{\eps F_p}{2}. \label{eq:exp6}
		\end{eqnarray}
		By \eqref{eq:exp1}, \eqref{eq:exp3} and \eqref{eq:exp6}, we have
		\begin{eqnarray*}
			\E[X] = \E[Y] - \E[Y-X] \geq \left(1 - 2\eta_p - \frac{\eps}{2} \right) F_p. 
		\end{eqnarray*}
	\end{proof}
	
	Finally, we prove the two claims.
	\begin{proof}[Proof of Claim~\ref{claim:small}]
		Let $n_\tau$ be the number of distinct elements in $\tau$, and $\{V_1, \ldots, V_{n_\tau}\}$ be the set of cliques in $G^\tau$. Note that
		\begin{eqnarray}
			\sum_{i \in H} (d_i^{\tau})^{p-1} &=& \sum_{\kappa \in [k]} \sum_{j \in [n_\tau]} \sum_{i \in \kV{\kappa} \cap V_j} (d_i^{\tau})^{p-1} \cdot \mathbf{1}\left\{ \frac{d_i^{\tau}}{\dl{i}[\tau]} > \alpha \right\} \label{eq:claim1.1}  \\
			&=& \sum_{\kappa \in [k]} \sum_{j \in [n_\tau]} \sum_{i \in \kV{\kappa} \cap V_j} \abs{V_j}^{p-1} \cdot \mathbf{1}\left\{ \frac{\abs{V_j}}{\abs{\kV{\kappa} \cap V_j}} > \alpha \right\} \label{eq:claim1.2} \\
			&=& \sum_{\kappa \in [k]} \sum_{j \in [n_\tau]} \abs{V_j}^{p-1} \cdot \abs{\kV{\kappa} \cap V_j} \cdot \mathbf{1}\left\{ \frac{\abs{V_j}}{\abs{\kV{\kappa} \cap V_j}} > \alpha \right\} \nonumber \\
			&<& \sum_{\kappa \in [k]} \sum_{j \in [n_\tau]} \frac{\abs{V_j}^p}{\alpha} = \sum_{\kappa \in [k]} \frac{F_p}{\alpha} = \frac{k F_p}{\alpha}, \nonumber
		\end{eqnarray}
		where from \eqref{eq:claim1.1} to \eqref{eq:claim1.2}, we use the fact that if $i \in V_j \cap \kV{\kappa}$, then $d_i^{\tau} = \abs{B_i^{\tau}} = \abs{V_j}$ and $\dl{i}[\tau] = \abs{B_i^{\tau} \cap \kV{\kappa}} = \abs{V_j \cap \kV{\kappa}}$.
	\end{proof}
	
	\begin{proof}[Proof of Claim~\ref{claim:large}]
		For any $i \in L$ such that $\left(\frac{d_i}{\dl{i}}\right)^{p-1} > \theta$, we have
		\begin{eqnarray}
			\theta &<& \frac{(d_i)^{p-1}}{(\dl{i})^{p-1}} \label{eq:claim2.0}\\ 
			&\leq& \frac{\abs{B_i^{\sigma} \cup B_i^{\tau}}^{p-1}}{\abs{B_i^{\sigma} \cap B_i^{\tau} \cap \kV{\kappa(i)}}^{p-1}} \label{eq:claim2.1} \\
			&=& \frac{(d_i^{\tau})^{p-1} + \left(\abs{B_i^{\sigma} \cup B_i^{\tau}}^{p-1} - (d_i^{\tau})^{p-1}\right)}{(\dl{i}[\tau])^{p-1} - \left((\dl{i}[\tau])^{p-1} - \abs{B_i^{\sigma} \cap B_i^{\tau} \cap \kV{\kappa(i)}}^{p-1} \right)}, \label{eq:claim2.2}
		\end{eqnarray}
		where from 	\eqref{eq:claim2.0} to	\eqref{eq:claim2.1} we use $d_i = \abs{B_i^{\sigma}} \leq \abs{B_i^{\sigma} \cup B_i^{\tau}}$ and $\dl{i} = \abs{B_i^{\sigma} \cap \kV{\kappa(i)}} \geq \abs{B_i^{\sigma} \cap B_i^{\tau} \cap \kV{\kappa(i)}}$.
		
		Rearranging the terms in \eqref{eq:claim2.2} gives
		\begin{eqnarray}
			\theta (\dl{i}[\tau])^{p-1} - (d_i^{\tau})^{p-1}
			&\leq& \left(\abs{B_i^{\sigma} \cup B_i^{\tau}}^{p-1} - (d_i^{\tau})^{p-1}\right) + \theta \left((\dl{i}[\tau])^{p-1} - \abs{B_i^{\sigma} \cap B_i^{\tau} \cap \kV{\kappa(i)}}^{p-1} \right) \nonumber \label{eq:claim2.3} \\
			&\leq& \left(\abs{B_i^{\sigma} \cup B_i^{\tau}}^{p-1} - (d_i^{\tau})^{p-1}\right) + \theta \left((d_i^{\tau})^{p-1}- \abs{B_i^{\sigma} \cap B_i^{\tau}}^{p-1} \right) \nonumber \label{eq:claim2.4} \\
			&\leq& \theta \left(\abs{B_i^{\sigma} \cup B_i^{\tau}}^{p-1} - \abs{B_i^{\sigma} \cap B_i^{\tau}}^{p-1}\right), \label{eq:claim2.5}
		\end{eqnarray}
		where 
		in the second inequality, we use the fact that for any $p \ge 2$ and any three sets $A,B,C$ such that $B \subseteq A$, we always have $\abs{A}^{p-1} - \abs{B}^{p-1} \geq \abs{A \cap C}^{p-1} - \abs{B \cap C}^{p-1}$.  The last inequality holds since $\theta \ge 1$.
		
		On the other hand, since $i \in L$ and $\alpha = \left(\frac{\theta}{2}\right)^{\frac{1}{p-1}}$, we have
		\begin{equation}
			\theta (\dl{i}[\tau])^{p-1} - (d_i^{\tau})^{p-1} \geq \theta \left(\frac{d_i^{\tau}}{\alpha}\right)^{p-1} - (d_i^{\tau})^{p-1} = (d_i^{\tau})^{p-1}. \label{eq:claim2.6}
		\end{equation}
		
		Combining \eqref{eq:claim2.5} and \eqref{eq:claim2.6}, we have for any $i \in L$ such that $\left(\frac{d_i}{\dl{i}}\right)^{p-1} > \theta$,
		\begin{equation*}
			(d_i^{\tau})^{p-1} \leq \theta \left(\abs{B_i^{\sigma} \cup B_i^{\tau}}^{p-1} - \abs{B_i^{\sigma} \cap B_i^{\tau}}^{p-1}\right).
		\end{equation*}
		It follows that
		\begin{eqnarray*}
			\sum_{i \in L} (d_i^{\tau})^{p-1} \cdot \mathbf{1}\left\{ \left(\frac{d_i}{\dl{i}}\right)^{p-1} > \theta \right\} \leq \sum_{i \in [m]} \theta \left(\abs{B_i^{\sigma} \cup B_i^{\tau}}^{p-1} - \abs{B_i^{\sigma} \cap B_i^{\tau}}^{p-1}\right) = \theta \eta_p F_p.
		\end{eqnarray*}
	\end{proof}
	
	\begin{lemma} \label{lem:FpDist-var}
		Let $\Bar{X} = \frac{1}{t} \sum_{i \in [t]} X_i$, where $X_i$ is defined as in Lemma~\ref{lem:FpDist-exp} for the $i$-th sample in $S$. Then with probability at least $\frac{2}{3}$,
		$
		(1 - 2\eta_p - \eps) F_p \leq \Bar{X} \leq (1 + \eta_p + \eps) F_p.
		$
	\end{lemma}
	
	\begin{proof}
		We first bound the variance of $X_i$.
		\begin{eqnarray*}
			\Var[X_i] &\leq& \E[X_i^2] \\
			&\leq& \theta \cdot D_{\loc} \cdot \E[X_i] \quad \text{(since $X_i \le  \theta \cdot D_{\loc}$)}\\
			&=& \theta \cdot \left( \sum_{i \in [m]} (\dl{i})^{p-1} \right) \cdot \E[X_i]  \le \theta \cdot \left( \sum_{i \in [m]} (d_i)^{p-1} \right) \cdot \E[X_i] \\
			&\leq& \theta \cdot (1 + \eta_p) F_p \cdot \E[X_i] \quad \text{(by \eqref{eq:exp3})}  \\
			&\leq& \theta \cdot \frac{1 + \eta_p}{1 - 2 \eta_p - \frac{\eps}{2}} \cdot \E[X_i]^2 \quad \text{(by Lemma~\ref{lem:FpDist-exp})} \\
			&\leq& 6\theta \cdot \E[X_i]^2.  \quad \text{(by assumption $\eta_p \leq {\eps^p}/({4^{p+1} \cdot k^{p-1}})$)}
		\end{eqnarray*}
		By Chebyshev's inequality,
		\begin{eqnarray*}
			\Pr\left[ \abs{\Bar{X} - \E[\Bar{X}]} > \frac{\eps}{2} \cdot \E[\Bar{X}] \right] \leq \frac{4}{t \eps^2}\cdot \frac{\Var(X_1)}{\E[X_1]^2} \leq \frac{24 \theta}{t \eps^2} = \frac{1}{3}.
		\end{eqnarray*}
		By Lemma~\ref{lem:FpDist-exp}, we have with probability $\frac{2}{3}$,
		\begin{eqnarray*}
			\Bar{X} &\leq& \left(1 + \frac{\eps}{2} \right)(1 + \eta_p) F_p \leq (1 + \eps + \eta_p) F_p, \quad \text{and} \\
			\Bar{X} &\geq& \left(1 - \frac{\eps}{2} \right)\left(1 - 2\eta_p - \frac{\eps}{2}\right) F_p \geq (1  - \eps - 2\eta_p) F_p,
		\end{eqnarray*}
		which completes the proof.
	\end{proof}
	
	Finally, Theorem~\ref{thm:FpDist} follows directly from Lemma~\ref{lem:FpDist-var}.  Note that we can always apply the median trick to reduce the error probability of Algorithm~\ref{alg:FpDist} to $0.01$ without changing the asymptotic communication complexity and the round complexity.  
	
	\begin{remark}
		We note that the lower bound of $\Omega\left(\frac{k^{p-1}}{\eps^2}\right)$ bits for $F_p$ in the coordinator model with noiseless data \cite{WZ12} also holds in the presence of noise. Consequently, in terms of $k$, the communication cost in Theorem~\ref{thm:FpDist} is only a factor of $k$ from optimal.
	\end{remark}
	
	\begin{remark}
		\label{rem:phase-transition}
		In our lower bound construction, the input of $F_p$, obtained from a YES instance of \kDISJ, has mismatch ambiguity at most $\frac{c_2\eps}{k}$ for a constant $c_2 = 10p$.  For $p = 2$, the ambiguity threshold in Theorem~\ref{thm:FpDist} is up to $\frac{c_1\eps^2}{k}$ for a constant $c_1 = \frac{1}{4^{p+1}}$.  Thus, the communication cost phase transition for $F_2$ happens in the range $[ \frac{c_1\eps^2}{k},   \frac{c_2\eps}{k}]$.
	\end{remark}

	\section{Concluding Remarks}
	In this work, we introduce the parameter {\em $F_p$-mismatch-ambiguity} and leverage it to design sublinear algorithms for the $F_p$ problem on noisy data in the data stream model and the coordinator model. Several directions remain open following this work. First, there are still gaps in $\eps$ between our upper and lower bounds that need to be closed. 
	
	Second, in the coordinator model, it would be interesting to obtain a complete characterization of the tradeoff between mismatch ambiguity and communication cost. In particular, is the communication cost phase transition mentioned in Remark~\ref{rem:phase-transition} smooth or abrupt?
	
	Third, it is worth exploring a setting where the similarity oracle $f$, given two items $\sigma_i$ and $\sigma_j$, outputs a real value in $[0, 1]$ rather than a Boolean value in $\{0, 1\}$ that indicates whether they are similar or dissimilar. By convention, we set $f(i, i) = 1$ for all $i \in [m]$. Define $g: [m] \times [m] \to \{0, 1\}$ such that $g(i, j) = 1$ if $\sigma_i$ and $\sigma_j$ correspond to the same ground truth element, and $g(i, j) = 0$ otherwise. We can then generalize the definition of $F_p$-mismatch-ambiguity as follows:
	\begin{equation*}
		\eta_p^{\text{wt}}(\sigma, \tau) = \frac{1}{F_p(\tau)}\sum_{i \in [m]} \left[\left(\sum_{j \in [m]}\max\{f(i, j), g(i, j)\}\right)^{p-1} - \left(\sum_{j \in [m]}\min\{f(i, j), g(i, j)\}\right)^{p-1}\right].
	\end{equation*}
	We leave open the question of whether efficient streaming and distributed algorithms can be developed to estimate $F_p$ when using $\eta_p^{\text{wt}}$ as a parameter in the approximation ratio.
	
	Finally, it would be interesting to study other statistical problems under our mismatch ambiguity framework.

	
\bibliographystyle{plain}
\bibliography{paper}

@inproceedings{ZH25,
	author       = {Qin Zhang and
	Mohsen Heidari},
	editor       = {Sudeepa Roy and
	Ahmet Kara},
	title        = {Quantum Data Sketches},
	booktitle    = {ICDT},
	series       = {LIPIcs},
	volume       = {328},
	pages        = {16:1--16:19},
	publisher    = {Schloss Dagstuhl - Leibniz-Zentrum f{\"{u}}r Informatik},
	year         = {2025}
}

@inproceedings{JMM+07,
	author       = {T. S. Jayram and
	Andrew McGregor and
	S. Muthukrishnan and
	Erik Vee},
	editor       = {Leonid Libkin},
	title        = {Estimating statistical aggregates on probabilistic data streams},
	booktitle    = {PODS},
	pages        = {243--252},
	publisher    = {{ACM}},
	year         = {2007}
}

@inproceedings{JKV07,
	author       = {T. S. Jayram and
	Satyen Kale and
	Erik Vee},
	editor       = {Nikhil Bansal and
	Kirk Pruhs and
	Clifford Stein},
	title        = {Efficient aggregation algorithms for probabilistic data},
	booktitle    = {SODA},
	pages        = {346--355},
	publisher    = {{SIAM}},
	year         = {2007}
}

@inproceedings{CG07,
	author       = {Graham Cormode and
	Minos N. Garofalakis},
	editor       = {Chee Yong Chan and
	Beng Chin Ooi and
	Aoying Zhou},
	title        = {Sketching probabilistic data streams},
	booktitle    = {SIGMOD},
	pages        = {281--292},
	publisher    = {{ACM}},
	year         = {2007}
}

@article{JYC+08,
	author       = {Cheqing Jin and
	Ke Yi and
	Lei Chen and
	Jeffrey Xu Yu and
	Xuemin Lin},
	title        = {Sliding-window top-k queries on uncertain streams},
	journal      = {Proc. {VLDB} Endow.},
	volume       = {1},
	number       = {1},
	pages        = {301--312},
	year         = {2008}
}

@inproceedings{ZLY08,
	author       = {Qin Zhang and
	Feifei Li and
	Ke Yi},
	editor       = {Jason Tsong{-}Li Wang},
	title        = {Finding frequent items in probabilistic data},
	booktitle    = {SIGMOD},
	pages        = {819--832},
	publisher    = {{ACM}},
	year         = {2008}
}

@inproceedings{HXZ+25,
	author       = {Zengfeng Huang and
	Zhongzheng Xiong and
	Xiaoyi Zhu and
	Zhewei Wei},
	editor       = {Michal Kouck{\'{y}} and
	Nikhil Bansal},
	title        = {Simple and Optimal Algorithms for Heavy Hitters and Frequency Moments
	in Distributed Models},
	booktitle    = {Proceedings of the 57th Annual {ACM} Symposium on Theory of Computing,
	{STOC} 2025, Prague, Czechia, June 23-27, 2025},
	pages        = {371--382},
	publisher    = {{ACM}},
	year         = {2025}
}

@inproceedings{LW13,
	author       = {Yi Li and
	David P. Woodruff},
	editor       = {Prasad Raghavendra and
	Sofya Raskhodnikova and
	Klaus Jansen and
	Jos{\'{e}} D. P. Rolim},
	title        = {A Tight Lower Bound for High Frequency Moment Estimation with Small
	Error},
	booktitle    = {APPROX-RANDOM},
	series       = {Lecture Notes in Computer Science},
	volume       = {8096},
	pages        = {623--638},
	publisher    = {Springer},
	year         = {2013}
}

@article{CMY+12,
	author       = {Graham Cormode and
	S. Muthukrishnan and
	Ke Yi and
	Qin Zhang},
	title        = {Continuous sampling from distributed streams},
	journal      = {J. {ACM}},
	volume       = {59},
	number       = {2},
	pages        = {10:1--10:25},
	year         = {2012}
}

@inproceedings{WZ12,
	author       = {David P. Woodruff and
	Qin Zhang},
	editor       = {Howard J. Karloff and
	Toniann Pitassi},
	title        = {Tight bounds for distributed functional monitoring},
	booktitle    = {STOC},
	pages        = {941--960},
	publisher    = {{ACM}},
	year         = {2012}
}

@inproceedings{Feige04,
	author       = {Uriel Feige},
	editor       = {L{\'{a}}szl{\'{o}} Babai},
	title        = {On sums of independent random variables with unbounded variance, and
	estimating the average degree in a graph},
	booktitle    = {Proceedings of the 36th Annual {ACM} Symposium on Theory of Computing,
	Chicago, IL, USA, June 13-16, 2004},
	pages        = {594--603},
	publisher    = {{ACM}},
	year         = {2004}
}

@inproceedings{EKM+24,
	author       = {Hossein Esfandiari and
	Praneeth Kacham and
	Vahab Mirrokni and
	David P. Woodruff and
	Peilin Zhong},
	editor       = {Bojan Mohar and
	Igor Shinkar and
	Ryan O'Donnell},
	title        = {Optimal Communication Bounds for Classic Functions in the Coordinator
	Model and Beyond},
	booktitle    = {STOC},
	pages        = {1911--1922},
	publisher    = {{ACM}},
	year         = {2024}
}

@article{Zhang25,
	author       = {Qin Zhang},
	title        = {Robust Statistical Analysis on Streaming Data with Near-Duplicates
	in General Metric Spaces},
	journal      = {Proc. {ACM} Manag. Data},
	volume       = {3},
	number       = {2},
	pages        = {111:1--111:25},
	year         = {2025}
}

@inproceedings{CKS03,
	author       = {Amit Chakrabarti and
	Subhash Khot and
	Xiaodong Sun},
	title        = {Near-Optimal Lower Bounds on the Multi-Party Communication Complexity
	of Set Disjointness},
	booktitle    = {CCC},
	pages        = {107--117},
	publisher    = {{IEEE} Computer Society},
	year         = {2003}
}

@inproceedings{BO13,
	author       = {Vladimir Braverman and
	Rafail Ostrovsky},
	editor       = {Prasad Raghavendra and
	Sofya Raskhodnikova and
	Klaus Jansen and
	Jos{\'{e}} D. P. Rolim},
	title        = {Generalizing the Layering Method of Indyk and Woodruff: Recursive
	Sketches for Frequency-Based Vectors on Streams},
	booktitle    = {APPROX-RANDOM},
	series       = {Lecture Notes in Computer Science},
	volume       = {8096},
	pages        = {58--70},
	publisher    = {Springer},
	year         = {2013}
}

@inproceedings{PVZ12,
	author       = {Jeff M. Phillips and
	Elad Verbin and
	Qin Zhang},
	editor       = {Yuval Rabani},
	title        = {Lower bounds for number-in-hand multiparty communication complexity,
	made easy},
	booktitle    = {SODA},
	pages        = {486--501},
	publisher    = {{SIAM}},
	year         = {2012}
}

@article{BYJ+04,
	author       = {Ziv Bar{-}Yossef and
	T. S. Jayram and
	Ravi Kumar and
	D. Sivakumar},
	title        = {An information statistics approach to data stream and communication
	complexity},
	journal      = {J. Comput. Syst. Sci.},
	volume       = {68},
	number       = {4},
	pages        = {702--732},
	year         = {2004}
}

@inproceedings{Zhang15,
	author    = {Qin Zhang},
	title     = {Communication-Efficient Computation on Distributed Noisy Datasets},
	booktitle = {SPAA},
	pages     = {313--322},
	year      = {2015}
}

@book{HSW07,
	title={Data quality and record linkage techniques},
	author={Herzog, Thomas N and Scheuren, Fritz J and Winkler, William E},
	volume={1},
	year={2007},
	publisher={Springer}
}

@article{EIV07,
	author    = {Ahmed K. Elmagarmid and
	Panagiotis G. Ipeirotis and
	Vassilios S. Verykios},
	title     = {Duplicate Record Detection: A Survey},
	journal   = {IEEE Trans. Knowl. Data Eng.},
	volume    = {19},
	number    = {1},
	year      = {2007},
	pages     = {1-16}
}

@inproceedings{KSS06,
	title={Record linkage: similarity measures and algorithms},
	author={Koudas, Nick and Sarawagi, Sunita and Srivastava, Divesh},
	booktitle={SIGMOD},
	pages={802--803},
	year={2006},
	organization={ACM}
}

@article{DN09,
	title={Data fusion: resolving data conflicts for integration},
	author={Dong, Xin Luna and Naumann, Felix},
	journal={Proceedings of the VLDB Endowment},
	volume={2},
	number={2},
	pages={1654--1655},
	year={2009},
	publisher={VLDB Endowment}
}

@inproceedings{CZ16,
	author       = {Di Chen and
	Qin Zhang},
	editor       = {Fatma {\"{O}}zcan and
	Georgia Koutrika and
	Sam Madden},
	title        = {Streaming Algorithms for Robust Distinct Elements},
	booktitle    = {SIGMOD},
	pages        = {1433--1447},
	publisher    = {{ACM}},
	year         = {2016}
}

@inproceedings{CZ18,
	author       = {Jiecao Chen and
	Qin Zhang},
	editor       = {Jan Van den Bussche and
	Marcelo Arenas},
	title        = {Distinct Sampling on Streaming Data with Near-Duplicates},
	booktitle    = {PODS},
	pages        = {369--382},
	publisher    = {{ACM}},
	year         = {2018}
}

@article{AMS99,
	author    = {Noga Alon and
	Yossi Matias and
	Mario Szegedy},
	title     = {The Space Complexity of Approximating the Frequency Moments},
	journal   = {J. Comput. Syst. Sci.},
	volume    = {58},
	number    = {1},
	pages     = {137--147},
	year      = {1999}
}

@inproceedings{IW05,
	author    = {Piotr Indyk and
	David P. Woodruff},
	title     = {Optimal approximations of the frequency moments of data streams},
	booktitle = {STOC},
	pages     = {202--208},
	year      = {2005}
}

@inproceedings{BGK+06,
	author    = {Lakshminath Bhuvanagiri and
	Sumit Ganguly and
	Deepanjan Kesh and
	Chandan Saha},
	title     = {Simpler algorithm for estimating frequency moments of data streams},
	booktitle = {SODA},
	pages     = {708--713},
	year      = {2006}
}

@inproceedings{MW10,
	author    = {Morteza Monemizadeh and
	David P. Woodruff},
	title     = {1-Pass Relative-Error L\({}_{\mbox{p}}\)-Sampling with Applications},
	booktitle = {SODA},
	pages     = {1143--1160},
	year      = {2010}
}

@inproceedings{AKO11,
	author    = {Alexandr Andoni and
	Robert Krauthgamer and
	Krzysztof Onak},
	title     = {Streaming Algorithms via Precision Sampling},
	booktitle = {FOCS},
	pages     = {363--372},
	year      = {2011}
}

@article{BO10,
	author    = {Vladimir Braverman and
	Rafail Ostrovsky},
	title     = {Recursive Sketching For Frequency Moments},
	journal   = {CoRR},
	volume    = {abs/1011.2571},
	year      = {2010}
}

@inproceedings{Andoni17,
	author    = {Alexandr Andoni},
	title     = {High frequency moments via max-stability},
	booktitle = {ICASSP},
	pages     = {6364--6368},
	year      = {2017}
}

@article{Ganguly11,
	author    = {Sumit Ganguly},
	title     = {Polynomial Estimators for High Frequency Moments},
	journal   = {CoRR},
	volume    = {abs/1104.4552},
	year      = {2011}
}

@inproceedings{Woodruff04,
	author    = {David P. Woodruff},
	title     = {Optimal space lower bounds for all frequency moments},
	booktitle = {SODA},
	pages     = {167--175},
	year      = {2004}
}

@article{Ganguly12,
	author    = {Sumit Ganguly},
	title     = {A Lower Bound for Estimating High Moments of a Data Stream},
	journal   = {CoRR},
	volume    = {abs/1201.0253},
	year      = {2012}
}

@inproceedings{GW18,
	author    = {Sumit Ganguly and
	David P. Woodruff},
	editor    = {Ioannis Chatzigiannakis and
	Christos Kaklamanis and
	D{\'{a}}niel Marx and
	Donald Sannella},
	title     = {High Probability Frequency Moment Sketches},
	booktitle = {ICALP},
	series    = {LIPIcs},
	volume    = {107},
	pages     = {58:1--58:15},
	publisher = {Schloss Dagstuhl - Leibniz-Zentrum f{\"{u}}r Informatik},
	year      = {2018}
}

@inproceedings{BZ25,
	author       = {Mark Braverman and
	Or Zamir},
	title        = {Optimality of Frequency Moment Estimation},
	booktitle    = {STOC},
	pages        = {360--370},
	year         = {2025}
}

@inproceedings{CMY08,
	author       = {Graham Cormode and
	S. Muthukrishnan and
	Ke Yi},
	title        = {Algorithms for distributed functional monitoring},
	booktitle    = {SODA},
	pages        = {1076--1085},
	year         = {2008}
}

@inproceedings{KVW14,
	author       = {Ravi Kannan and
	Santosh S. Vempala and
	David P. Woodruff},
	title        = {Principal Component Analysis and Higher Correlations for Distributed
	Data},
	booktitle    = {COLT},
	series       = {{JMLR} Workshop and Conference Proceedings},
	volume       = {35},
	pages        = {1040--1057},
	year         = {2014}
}

@article{JW23,
	author       = {Rajesh Jayaram and
	David P. Woodruff},
	title        = {Towards Optimal Moment Estimation in Streaming and Distributed Models},
	journal      = {{ACM} Trans. Algorithms},
	volume       = {19},
	number       = {3},
	pages        = {27:1--27:35},
	year         = {2023}
}

	%\appendix
	%\input{appendix}
	
\end{document}